\documentclass[letterpaper]{vldb}

\usepackage{times}				% fonts change, saves abount 0.5 pages per 8 page ACM template (June 2010)
\usepackage{graphicx}     		% can better scale and rotate graphics than graphic package
\usepackage{multirow}      		% allows to have table cells on more than one row with "\multirow"
\usepackage{amssymb}			% allows special symbols like "\nexists"
\usepackage{latexsym}			% ? Dan
\usepackage{alltt}				% ? Dan
\usepackage{amsgen}				% ? Dan
\usepackage{amstext}			% ? Dan
\usepackage{amscd}				% ? Dan

\usepackage{xspace}            % used for macros; inserts empty space if required
\usepackage[caption=false]{subfig}            % replaces old "subfigure package; calls itself caption; 'caption=false' prevents screwing up IEEE class defined caption styles
\usepackage{enumitem}               % allows customization of lists, enumerations, descriptions; used for "list of videos" / seems not to work with IEEE

\usepackage{color}			% to use visible comments, margin index comments, white phantom "inputs" in algorithm2e
\usepackage{colortbl}		% allows colors in tables

\usepackage{tabularx}		% used for correct indentation in SQL translations. table environment that allows to span a whole line

\usepackage{aliascnt}  	% ``hyperrefÕs \autoref command does not work well with theorems that share a counter:
\usepackage[novbox]{pdfsync}		% allows jump between source code and pdf; [novbox] parameter important to harmonize with tabularx package

\usepackage{float}                 % Should be included for hyperref according to Google Groups + algorithm + hyperref + float; in fact messes links to algorithms up
\usepackage{ifpdf}
    \ifpdf                                                                  % if PDFLATEX
        \usepackage[breaklinks=true]{hyperref}
        \hypersetup{
            colorlinks=false,               % false makes boxes for viewing, nothing=better for printing
            citebordercolor={0 1 0},        % green                     % RGB necessary for "border" values: red={1 0 0}, green={0 1 0}, blue= {0 0 1}
            urlbordercolor={0 0 1},         % blue                      % yellow={1 1 0}, cyan(greenish light blue)={1 0 1}
            linkbordercolor={1 0.2 0.2},    % red                       % warm orange yellow={1 0.65 0}, light red={1 0.5 0.5}
            pdfborder={0 0 1.2},            % border thickness
        }
    \else
        \usepackage[hypertex, breaklinks=true]{hyperref}       % if DVI; "hypertex" makes links in DVI
        \hypersetup{
            colorlinks=true,               % false makes boxes for viewing, nothing=better for printing
            anchorcolor=yellow,
            linkcolor=green,
            filecolor=green,
            urlcolor=blue,
            citecolor=red,
        }
    \fi
    \hypersetup{                                                            % for both, DVI and PDF
        pdfpagemode=UseNone,               % pdfpagemode=UseOutlines
        pdfstartview=Fit,
        pdfpagelayout=SinglePage,       % pdfpagelayout=SinglePage
        bookmarks=true,
        bookmarksnumbered=true,
        bookmarksopenlevel=1,            % ??? only partly expanded
        bookmarkstype=toc,
        pdftex=true,
        unicode,
        plainpages=false,           % necessary for correct backref
        hypertexnames=true,         % necessary for correct backref [??? false might be necessary in order to get the hyper-links to the figures right]
        pdfpagelabels=true,         % necessary for correct backref
        hyperindex=true,
        naturalnames=false,                % necessary for not confusing e.g. roman page ii with later arabic 2; use LaTeX-computed names for links (default: false)
    }

\newtheorem{theorem}{Theorem}[section]          	% Theorem environment.

\newaliascnt{lemma}{theorem}				% 1 alias counter
\newtheorem{lemma}[lemma]{Lemma}              	% Lemma environment.
\aliascntresetthe{lemma}  					% 3 set
\newaliascnt{conjecture}{theorem}			% 1 alias counter
\aliascntresetthe{conjecture}  				% 3 set
\newaliascnt{remark}{theorem}				% 1 alias counter
              
\aliascntresetthe{remark}  					% 3 set
\newaliascnt{corollary}{theorem}			% 1 alias counter
\newtheorem{corollary}[corollary]{Corollary}      % Corollary environment.
\aliascntresetthe{corollary}  				% 3 set
\newaliascnt{definition}{theorem}			% 1 alias counter
\newtheorem{definition}[definition]{Definition}    % Definition environment.
\aliascntresetthe{definition}  				% 3 set
\newaliascnt{proposition}{theorem}			% 1 alias counter
\newtheorem{proposition}[proposition]{Proposition}  % proposition environment.
\aliascntresetthe{proposition}  				% 3 set
\newaliascnt{example}{theorem}			% 1 alias counter
\newtheorem{example}[example]{Example}  	% 2 environment.
\aliascntresetthe{example}  				% 3 set
\let\orgautoref\autoref                         		% \Autoref for the beginning of the sentence (cf. http://www.latex-community.org/viewtopic.php?f=4&p=1232 (June 2007))
\providecommand{\Autoref}[1]{%                  % comment marks "%" at each line are important, otherwise unnecessary space in result
    \def\equationautorefname{Equation}%	% also: do not leave any empty lines inside this paragraph! (Nov 2009)
    \def\figureautorefname{Figure}%
	\def\subfigureautorefname{Figure}%
    \def\lemmaautorefname{Lemma}%
    \def\conjectureautorefname{Conjecture}%
    \def\remarkautorefname{Remark}%   
    \def\propositionautorefname{Proposition}%
    \def\corollaryautorefname{Corollary}%
    \def\definitionautorefname{Definition}%
    \def\sectionautorefname{Section}%
    \def\subsectionautorefname{Section}%
    \def\subsubsectionautorefname{Section}%
    \def\exampleautorefname{Example}%
    \orgautoref{#1}%
}
\renewcommand{\autoref}[1]{%                    % \autoref inside the sentence to produce Fig., Eq.,etc
    \def\equationautorefname{Eq.}%		% no empty lines inside this paragraph! (Nov 2009)
    \def\figureautorefname{Fig.}%
    \def\subfigureautorefname{Fig.}%
    \def\lemmaautorefname{Lemma}%
    \def\conjectureautorefname{Conjecture}%
    \def\remarkautorefname{Remark}%
    \def\propositionautorefname{Prop.}%
    \def\corollaryautorefname{Corollary}%    
    \def\definitionautorefname{Def.}%
    \def\sectionautorefname{Sect.}%
    \def\subsectionautorefname{Sect.}%
    \def\subsubsectionautorefname{Section}%
    \def\exampleautorefname{Example}%
    \orgautoref{#1}%
}

\newcommand{\specificref}[2]{\hyperref[#2]{#1~\ref*{#2}}}			% allows specific keyword for autoref, e.g. Algorithm (Oct 2009)

\usepackage{scalefnt}       % Allows to relatively scale font sizes with "\scalefont{0.9}", works better than package "relsize"
\makeatletter               % Define a smaller 'leo' style for "\url" from "hyperref"(see: http://www.kronto.org/thesis/tips/url-formatting.html)
    \def\url@leostyle{
        \@ifundefined{selectfont}{\def\UrlFont{\sf}}{\def\UrlFont{\scalefont{0.9}\fontfamily{cmtt}\selectfont}}}
\makeatother
\usepackage{algorithm}					% alternative algorithm environment (June 2010)
\usepackage{algorithmic}				% 

\usepackage[algo2e, ruled, vlined]{algorithm2e}

\renewcommand{\epsilon}{\varepsilon}    % nicer epsilon symbol

\newcommand{\sql}[1]{\textup{\textsf{\scriptsize #1}}}    % textup for math environments
\newcommand{\sqll}[1]{{\fontfamily{cmtt}\selectfont{\textup{#1}}}}
\newcommand{\movie}[1]{``{#1}''}
\newcommand{\genre}[1]{#1}
\newcommand{\director}[1]{#1}

\newcommand{\introparagraph}[1]{\textbf{#1.}}        % define own new subsection type: noindent, bold (textsc)

\definecolor{gray}{rgb}{0.5,0.5,0.5}
\definecolor{niceblue}{rgb}{.8,.85,1}
\newcommand{\rewrite}{\leadsto}
\newcommand{\rewriteclosed}								% June 2010: Problem however: different behavior
{\vbox{													% inbetween variables or by itself
	\baselineskip=0pt 	
	\hbox{\hspace{+0.8mm} \vspace{-0.3em}\small{$^*$}}
	\hbox{\hspace{-0.0mm} \vspace{0.0em}$\leadsto$  \hspace{-0.1mm}}
}}
\newcommand{\weakening}{\multimap}
\newcommand{\weakeningclosed}
{\vbox{
	\baselineskip=0pt 	
	\hbox{\hspace{0.8mm} \vspace{-0.5em}\small{$^*$}}
	\hbox{\hspace{0.0mm} \vspace{0.0em}$\multimap$ \hspace{-0.1mm}}
}}

\newcommand{\enSymb}{\textup{n}}
\newcommand{\exSymb}{\textup{x}}

\newcommand{\en}[1]{{#1}^{\enSymb}}					% endogenous and exogenous tuples, relations (June 2010)
\newcommand{\ex}[1]{{#1}^{\exSymb}}

\newcommand{\datarule}{{\,:\!\!-\,}}			% Datalog command (March 2010)
\newcommand{\Var}{\textup{\textit{Var}}}      	% Var
\newcommand{\Adom}{\textup{\textit{Adom}}}		% Adom	
\newcommand{\sg}{\textit{sg}}					% subgoals

\newcommand{\whyso}{{{Why-So}}}
\newcommand{\whyno}{{{Why-No}}}

\usepackage{bigstrut}				% allows to make higher row heights in table

\renewenvironment{thebibliography}[1]{%
    \begin{oldthebibliography}{#1}%
        \footnotesize
        \setlength{\itemsep}{0.4ex}%
}%
{%
\end{oldthebibliography}%
}

  {\refstepcounter{theorem}\trivlist\item       %   Acts just like a theorem
  [\hskip\labelsep\bf\sf Policy Query \thetheorem]}%    %   environment, except that
  {\endtrivlist}                                %   body is typeset in
\newcommand{\set}[1]{\{#1\}}                    % Set (as in \set{1,2,3}).
\newcommand{\makeop}[2]                         % Macro to make new math syms.
  {\ifx#2.\def\next##1{}\else\escapechar=-1     %   Defines 2nd and subsequent
  \def\next##1{\escapechar=92\def#2{#1}}        %   args to expand to 1st arg,
  \expandafter\next\expandafter{\string#2}      %   with occurrences of "#1"
  \let\next\makeop\fi\next{#1}}                 %   replaced by cmd name.
\makeop{{\cal #1}}                              % Various calligraphic letters.
  \W  .             %
\def \var(#1){{\bf #1}}

\def\AddSpace#1{\ifcat#1a\ \fi#1} % if next is a letter, add a space

\newcommand{\silentreminder}[1]{}

\def \up(#1){[#1)}
\def \down(#1){(#1]}

\def \series(#1,#2){#1_1, \dots \; #1_{#2}}
\def \serieszero(#1,#2){#1_0, #1_1, \dots \; #1_{#2}}

\def \para(#1){{\vspace{1ex}\noindent\small\bf #1\hspace{1ex}}}
\def \myem(#1){{\vspace{1ex}\noindent\small\em #1\hspace{1ex}}}

\newcommand{\eat}[1]{}

\makeatletter 
\renewcommand{\@opargbegintheorem}[3]{%
    \parskip 0pt % GM July 2000 (for tighter spacing)
    \trivlist
    \item[%
        \hskip 10\p@
        \hskip \labelsep
        {\sc #1\ #2\             % 
   \setbox\@tempboxa\hbox{(#3)}  % 
        \ifdim \wd\@tempboxa>\z@ % 
            \hskip 0\p@\relax    % This number has to be reduced from 5 !!!
            \box\@tempboxa       % 
        \fi.}%
    ]
    \it
}
\makeatother

\newcommand{\bibpath}{causation} 		% Path to bibliography; ! just 1 brace as compared to graphicspath
\graphicspath{{figs/}}                       		% path to figures and graphics; the same directory would be "{./}"; "{figs/}"; ! needs two nested braces
        
\hypersetup{
	pdftitle={},
         pdfauthor={},
         pdfsubject={},
         pdfkeywords={}
	}
\begin{document}
\title{The Complexity of Causality and Responsibility for \\Query Answers and non-Answers}
\numberofauthors{1} 
\author{\alignauthor Alexandra Meliou \hspace{8mm} Wolfgang Gatterbauer \hspace{8mm}Katherine F. Moore \hspace{8mm} Dan Suciu \\
\affaddr{Department of Computer Science and Engineering,} \\
\affaddr{University of Washington, Seattle, WA, USA} \\
\email{\normalsize \texttt{ \{ameli,gatter,kfm,suciu\}@cs.washington.edu}}}  
\toappear{University of Washington CSE Technical Report}
\maketitle
\pagenumbering{arabic}      

\begin{abstract}
  An answer to a query has a well-defined lineage expression (alternatively called how-provenance)
  that explains how the answer was derived. Recent work has also
  shown how to compute the lineage of a non-answer to a
  query. However, the {\em cause} of an answer or non-answer is a more subtle notion and consists, in general, of only a fragment of the lineage. In this paper, we adapt Halpern,
  Pearl, and Chockler's recent definitions of causality and responsibility 
to define the \emph{causes of answers and non-answers to queries, and their degree of responsibility}. Responsibility captures the notion of degree of causality and serves to rank potentially many causes by their relative contributions to the effect.
Then, we study the
  complexity of computing causes and responsibilities for
  conjunctive queries. It is known that computing causes is NP-complete in
    general. Our first main result shows that all
  causes to conjunctive queries can be computed by a relational query
  which may involve negation. Thus, causality can be computed in
  PTIME, and very efficiently so. Next, we study computing responsibility.
Here, we prove that the complexity depends on
  the conjunctive query and demonstrate a dichotomy between PTIME and NP-complete cases.  For the
  PTIME cases, we give a non-trivial algorithm, consisting of a
  reduction to the max-flow computation problem. Finally, we prove
  that, even when it is in PTIME, responsibility is complete for
  \textsc{logspace}, implying that, unlike causality, it cannot be computed by
  a relational query.
\end{abstract}

\section{Introduction} 
\label{sec:introduction}
\looseness -1
When analyzing complex data sets, users are often interested in the
\emph{reasons} for surprising observations.  In a database context,
they would like to find the \emph{causes of answers or non-answers} to
their queries. For example, ``What caused my personalized newscast to
have more than 50 items today?''  Or, ``What caused my favorite
undergrad student to not appear on the Dean's list this year?''
Philosophers have debated for centuries various notions of causality,
and today it is still studied in philosophy, AI, and cognitive
science.  Understanding causality in a broad sense is of vital
practical importance, for example in determining legal responsibility
in multi-car accidents, in diagnosing malfunction of complex systems,
or scientific inquiry.  A formal, mathematical study of causality was
initiated by the recent work of Halpern and
Pearl~\cite{HalpernPearl:Cause2005} and Chockler and
Halpern~\cite{DBLP:journals/jair/ChocklerH04}, who gave mathematical
definitions of {\em causality} and its related notion of {\em degree
  of responsibility}.  These formal definitions lead to applications
in knowledge representation and model
checking~\cite{DBLP:journals/ai/EiterL02,Eiter:2006p86,DBLP:journals/jair/ChocklerH04}. In
this paper, we adapt the notions of causality and responsibility to
database queries, and study the complexity of computing the causes and
their responsibilities for answers and non-answers to
conjunctive queries.

\begin{figure}[t]
	\centering
		\includegraphics[scale=0.58]{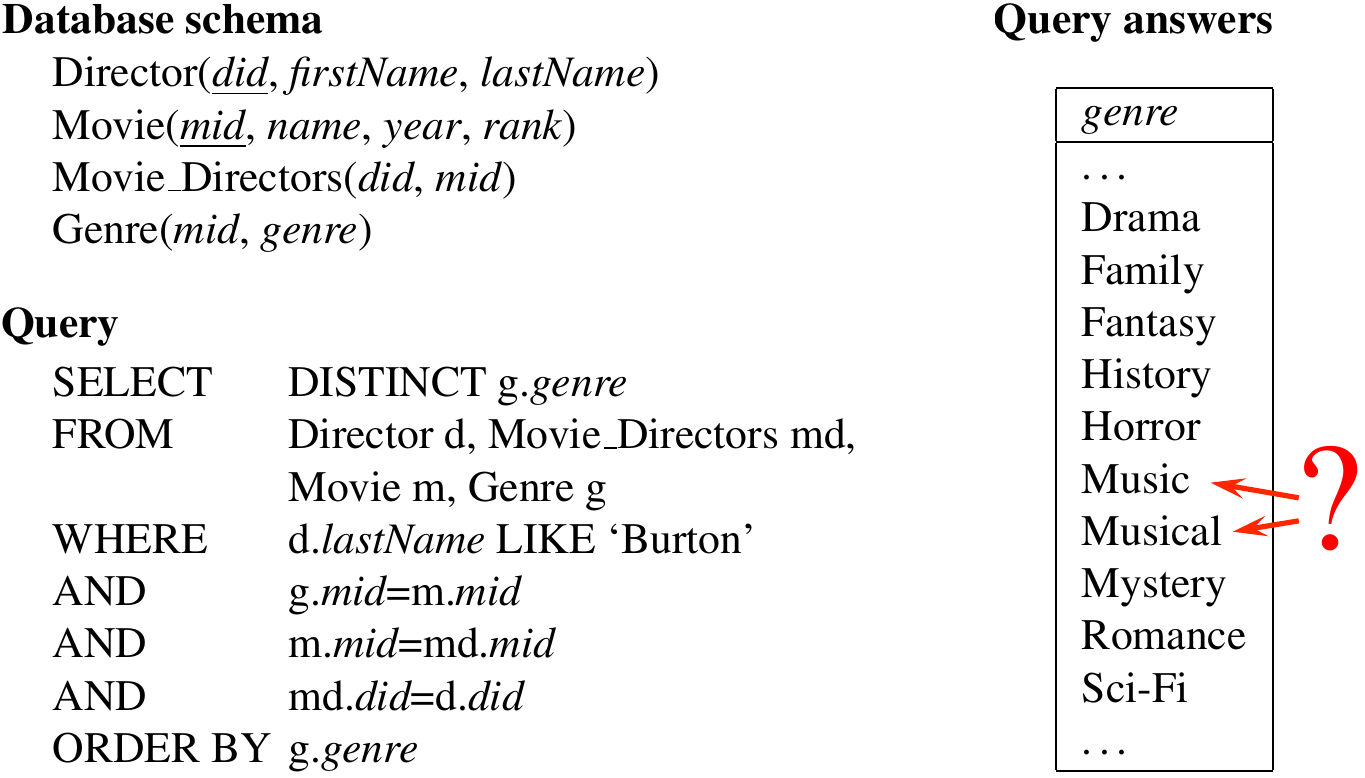}
		\vspace{-0.1in}
\caption{A SQL query returning the genres of all movies directed by
  Burton, on the IMDB dataset (\url{www.imdb.org}).  The famous
  director Tim Burton is known for dark, gothic themes, so the genres
  \genre{Fantasy} and \genre{Horror} are expected.  But the genres \genre{Music}
  and \genre{Musical} are quite surprising.  The goal of this paper is to
  find the \emph{causes for surprising query results}.}
	\label{fig:imdb}
	\vspace{-0.1in}
\end{figure}

\begin{figure*}[tb]
	\centering
	\begin{minipage}[t]{100mm}
		\subfloat[Lineage from \sqll{Directors} and \sqll{Movies} of the \genre{Musical} tuple]{
		\includegraphics[scale=0.48]{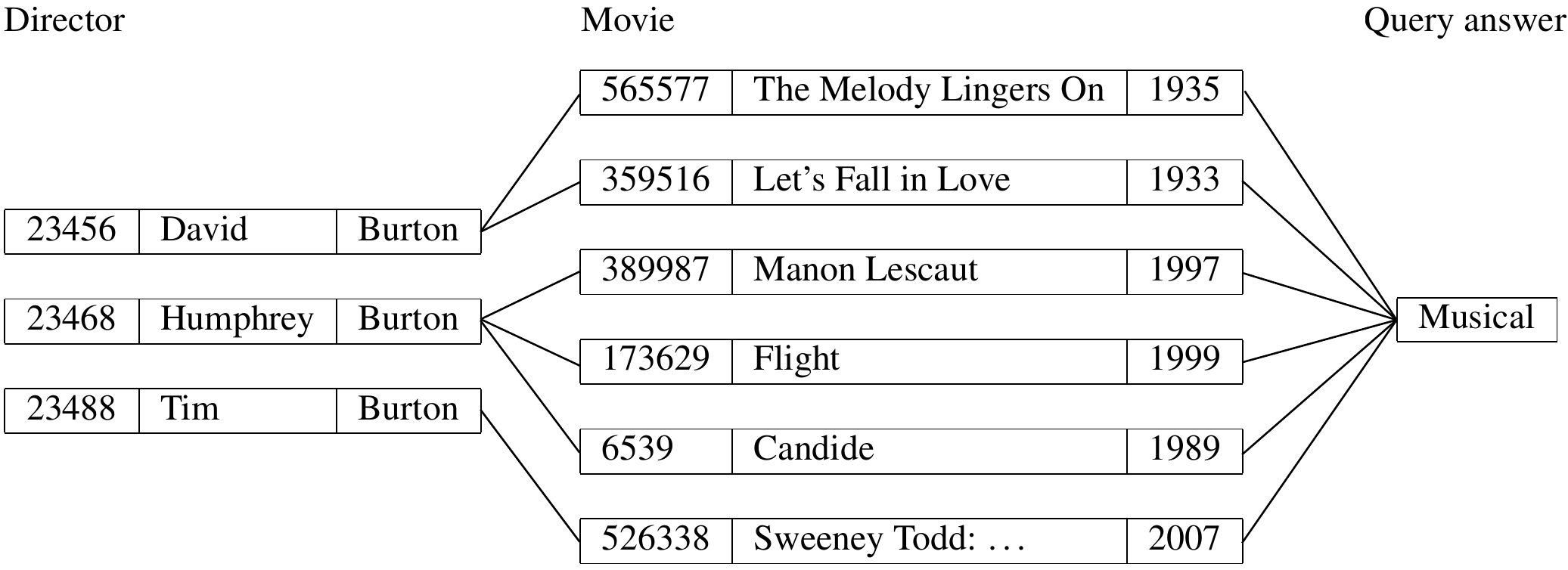}
		\label{fig:musicals-graph}}
	\end{minipage}
	\hspace{14mm}
	\begin{minipage}[t]{61mm}
	\subfloat[Responsibility rankings for \genre{Musical}]{
		{\scriptsize
		\setlength{\tabcolsep}{1.1mm}
		\renewcommand{\arraystretch}{1.3}
		\begin{tabular}[t]{|l|l|}
			\hline
			$\rho_t$ & Answer tuple \\
			\hline
			0.33 & Movie(526338, ``Sweeney Todd'', 2007)\\
			0.33 & Director(23456, David, Burton)\\
			0.33 & Director(23468, Humphrey, Burton)\\
			0.33 & Director(23488, Tim, Burton)\\
			0.25 & Movie(359516, ``Let's Fall in Love'', 1933)\\
			0.25 & Movie(565577, ``The Melody Lingers On'', 1935)\\
			0.20 & Movie(6539, ``Candide'', 1989)\\
			0.20 & Movie(173629, ``Flight'', 1999)\\
			0.20 & Movie(389987, ``Manon Lescaut'', 1997)\\
			\hline
		\end{tabular}
		\label{fig:musicals-rank}}}
	\end{minipage}
	\vspace{-0.1in}
	\caption{Lineage (a) and causes with their responsibilities (b) for the \sql{Musical} tuple 
	in Example~\ref{ex:imdb}.}
	\label{fig:musicals}
	\vspace{-0.1in}
\end{figure*}

\vspace{-1mm}
\begin{example}[IMDB]\label{ex:imdb}
    \looseness -1
  Tim Burton is an Oscar nominated director whose movies often include
  fantasy elements and dark, gothic themes. Examples of his work are
  \movie{Edward Scissorhands}, \movie{Beetlejuice} and the recent \movie{Alice in
  Wonderland}.  A user wishes to learn more about Burton's movies and
  queries the IMDB dataset to find out all genres of movies that he
  has directed (see~\autoref{fig:imdb}). \genre{Fantasy} and \genre{Horror}
  are quite expected categories.  But \genre{Music} and \genre{Musical} are
  surprising.  The user wishes to know the reason for these answers.
  Examining the lineage of a surprising answer is a first step towards
  finding its reason, but it is not sufficient: the combined lineage
  of the two categories consists of a total of 137 base tuples, which
  is overwhelming to the user.
\end{example}
\vspace{-1mm}

\looseness -1
Causality is related to provenance, yet it is a more
refined notion: Causality can answer questions like the one in our
example by returning the causes of query results ranked by their
\emph{degree of responsibility}.
Our starting point is Halpern and Pearl's definition of causality \cite{HalpernPearl:Cause2005},
from which we borrow three important concepts:

\looseness -1
(1) Partitioning of variables into {\em exogenous} and {\em
  endogenous}: Exogenous variables define a context determined by
external, unconcerned factors, deemed not to be possible causes, while
endogenous variables are the ones judged to affect the outcome and are
thus potential causes.  In a database setting, variables are tuples in
the database, and the first step is to partition them into exogenous
and endogenous.  For example, we may consider \sqll{Director} and
\sqll{Movie} tuples as endogenous and all others as exogenous.  The
classification into endogenous/exogenous is application-dependent, and
may even be chosen by the user at query time.  For example, if
erroneous data in the directors table is suspected, then only
\sqll{Director} may be declared endogenous; alternatively, the user
may choose only \sqll{Movie} tuples with \sqll{year}$>$\sqll{2008} to be
endogenous, for example in order to find recent, or under production
movies that may explain the surprising outputs to the query.  Thus,
the partition into endogenous and exogenous tuples is not restricted
to entire relations.  As a default, the user may start by declaring
all tuples in the database as endogenous, then narrow down.

\looseness -1 
(2) {\em Contingencies}: an endogenous tuple $t$ is a cause for the
observed outcome only if there is a hypothetical setting of the other
endogenous variables under which the addition/removal of $t$ causes
the observed outcome to change.  Therefore, in order to check that a
tuple $t$ is a cause for a query answer, one has to find a set of
endogenous tuples (called {\em contingency}) to remove from (or add
to) the database, such that the tuple
$t$ immediately affects the answer in the new state of the database.  In theory, in order to compute
the contingency one has to iterate over subsets of endogenous tuples.
Not surprisingly, checking causality is NP-complete in
general~\cite{DBLP:journals/ai/EiterL02}.  However, the first main
result in this paper is to show that the causality of conjunctive
queries can be determined in PTIME, and furthermore, all causes can
be computed by a relational query.

\looseness -1
(3) {\em Responsibility}, a notion first defined in \cite{DBLP:journals/jair/ChocklerH04}, measures the degree of causality as a
function of the size of the smallest contingency set.  In applications
involving large datasets, it is critical to rank the candidate
causes by their responsibility, because answers to complex queries may
have large lineages and large numbers of candidate causes.  In
theory, in order to compute the responsibility one has to iterate over
all contingency sets: not surprisingly, computing responsibility in
general is hard for $FP^{\Sigma_2^P(\log
  n)}$~\cite{DBLP:journals/jair/ChocklerH04}.\footnote{This is the class of functions computable by a
    poly-time Turing machine which makes $\log n$ queries to a
    $\Sigma_2^p$ oracle.}
However, our second main result, and at the same time the strongest result of this paper, is a \emph{dichotomy result for conjunctive
queries}: for {\em each} query without self-joins, either its responsibility can be
computed in PTIME in the size of the database (using a non-obvious
algorithm), or checking if it has a responsibility below a given value
is NP-hard.

\vspace{-1mm}
\begin{example}[IMDB continued]~\!\!\!\!
    \looseness -1
  Continuing \autoref{ex:imdb}, we show in \autoref{fig:musicals-rank}
  the causes for \genre{Musical} ranked by their responsibility score. (We
  explain in \autoref{sec:defining_causes_for_query_results} how these
  scores are computed.)  At the top of the list is the movie \movie{Sweeney
  Todd}, which {\em is}, indeed, the one and single musical movie
  directed by \director{Tim Burton}.  Thus, this tuple represents a surprising
  fact in the data of great interest to the user.  The next three
  tuples in the list are directors, whose last name is Burton.  These
  tuples too are of high interest to the user because they indicate
  that the query was ambiguous.  Equally interesting is to look at the
  bottom of the ranked list.  The movie \movie{Manon Lescaut} is made by
  \director{Humphrey Burton}, a far less known director specialized in musicals.
  Clearly, the movie itself is not an interesting explanation to the
  user; the interesting explanation is the director, showing that he
  happens to have the same last name, and indeed, the director is
  ranked higher while the movie is (correctly) ranked lower.  In our
  simple example \genre{Musical} has a small lineage, consisting of only
  ten tuples.  More typically, the
  lineage can be much larger (\genre{Music} has a lineage with 127 tuples), and it is critical to rank the potential
  causes by their degree of responsibility.
\end{example}
\vspace{-1mm}

\looseness -1
We start by adapting the Halpern and Pearl definition of causality
(HP from now on) to database queries, based on contingency sets.  We
define causality and responsibility both for {\whyso} queries
(``why did the query return this answer?'') and for {\whyno} 
queries (``why did the query not return this answer?'').  We then
prove two fundamental results.
First, we show that computing the causes to any conjunctive query can
be done in PTIME in the size of the database, i.e.\  query causality
has PTIME data complexity; by contrast, causality of arbitrary Boolean
expressions is NP-complete \cite{DBLP:journals/ai/EiterL02}.  In fact
we prove something stronger: the set of all causes can be retrieved by
a query expressed in First Order Logic (FO).  This has important
practical consequences, because it means that one can retrieve all
causes to a conjunctive query by simply running a certain SQL query.
In general, the latter cannot be a conjunctive query, but must have
one level of negation.  However, we show that if the user query has no
self joins {\em and} every table is either entirely endogenous or
entirely exogenous, then the {\whyso} causes can be retrieved by some
conjunctive query.  These results are summarized in
\autoref{fig:results}.

\looseness -1
Second, we give a dichotomy theorem for query responsibility. 
This is our strongest technical result with this paper. For
every conjunctive query without self-joins, one of the following
holds: either the responsibility can be computed in PTIME or it is
provably NP-hard.  In the first case, we give a quite non-obvious
algorithm for computing the degrees of responsibility using
Ford–Fulkerson's max flow algorithm.  We further show that one can
distinguish between the two cases by checking a property of the query
expression that we call {\em linearity}.  We also discuss conjunctive
queries with self-joins, and finally show that, in the case of {\whyno} 
causality, one can always compute responsibility in PTIME.  These
results are also summarized  in \autoref{fig:results}.

\begin{figure}[htb]
	\setlength{\tabcolsep}{1.1mm}
	\renewcommand{\arraystretch}{1.2}
	\centering
	\begin{tabular}{@{\hspace{0pt}}l@{\hspace{0pt}}llc|c|c|}
		\cline{5-6}
		& \multicolumn{3}{l|}{\textbf{\!\!Causality}} & \it{Why So?} & \it{Why No?}\\
		\cline{2-6}	
		\; & \multicolumn{3}{|l|}{\it{w/o SJ}} & PTIME (CQ) & \multirow{2}{*}{PTIME (FO)} \bigstrut\\
		\cline{2-5}
		& \multicolumn{3}{|l|}{\it{with SJ}}  & PTIME (FO) & \bigstrut\\
		\cline{2-6}
		& \multicolumn{3}{l}{}\\[0mm]
		\cline{5-6}
		& \multicolumn{3}{l|}{\!\!\textbf{Responsibility}} & { \it Why So?} & { \it Why No?}\\
		\cline{2-6}
		& \multicolumn{1}{|l|}{\multirow{2}{*}{\it w/o SJ}}  & {\it linear} 
			& 
			& PTIME & \bigstrut\\
		\cline{3-5}
		& \multicolumn{1}{|l|}{}  & {\it non-linear} & 
			& {NP-hard} &  PTIME \bigstrut\\
		\cline{2-5}
		& \multicolumn{3}{|l|}{\it with SJ} & NP-hard &\bigstrut\\
		\cline{2-6}
	\end{tabular}\label{fig:complexityRes_responsibility}
    \caption{Complexity of determining causality and responsibility for conjunctive queries. For queries with no self-joins we provide a complete dichotomy result. Queries with self-joins are NP-hard in general, but a similar dichotomy is not known.}\label{fig:results}
\end{figure}
\looseness -1
{\bf Causality and provenance:} Causality is related to \emph{lineage of query results},
such as why-provenance~\cite{DBLP:journals/tods/CuiWW00} or where-provenance~\cite{DBLP:conf/icdt/BunemanKT01}.
Recently, even explanations for non-answers have been described in
terms of
lineage~\cite{DBLP:journals/pvldb/HuangCDN08,DBLP:conf/sigmod/ChapmanJ09}.
We make use of this prior work because the first step in computing
causes and responsibilities is to determine the lineage of an answer
or non-answer to a query.  We note, however, that computing the
lineage of an answer is only the first step, and is not sufficient for
determining causality: causality needs to be established through a
contingency set, and is also accompanied by a degree (the
responsibility), which are both more difficult to compute than the lineage.

\introparagraph{Contributions and outline} Our three main contributions are:
\begin{itemize}[itemsep=1pt, parsep=1pt, topsep = 1pt]
\item We define {\whyso} and {\whyno}  causality and responsibility for conjunctive
  database queries (\autoref{sec:defining_causes_for_query_results}).
\item We prove that causality has PTIME data complexity for conjunctive queries
  (\autoref{sec:complexity_analysis_causality}).
\item We prove a dichotomy theorem for responsibility and conjunctive queries
  (\autoref{sec:complexity_analysis_responsibility}).
\end{itemize}
We review related work (\autoref{sec:related_work}) before we conclude 
(\autoref{sec:conclusions}).
All proofs are provided in the \hyperref[sec:appendix]{Appendix}.

\section{Query Cause and Responsibility} 
\label{sec:defining_causes_for_query_results}
\looseness -1
We assume a standard relational schema with relation names $R_1,
\ldots, R_k$.  We write $D$ for a database instance and $q$ for a
query.  We consider only conjunctive queries, unless otherwise
stated. A subset of tuples $\en{D} \subseteq D$ represents {\em
  endogenous tuples}; the complement $\ex{D} = D - \en{D}$ is called
the set of {\em exogenous tuples}.  For each relation $R_i$, we write
$\en{R}_i$ and $\ex{R}_i$ to denote the endogenous and exogenous
tuples in $R_i$ respectively.  If $\bar a$ is a tuple with the same
arity as the query's answer, then we write $D \models q(\bar{a})$ when
$\bar{a}$ is an answer to $q$ on $D$, and write $D \not\models
q(\bar{a})$ when $\bar{a}$ is a non-answer to $q$ on $D$.

\begin{definition}[Causality] \label{def:causality}\!\!  Let $t\!\in\!
  \en{D}$ be an endogenous tuple, and $\bar{a}$ a possible answer for
  $q$.
	\begin{itemize}[itemsep=1pt, parsep=1pt, topsep = 1pt]
        \item $t$ is called a {\em counterfactual cause} for $\bar{a}$
          in $D$ if $D \models q(\bar a)$ and $D - \{ t\}\not\models
          q(\bar{a})$
        \item $t\in D$ is called an {\em actual cause} for $\bar{a}$
          if there exists a set $\Gamma \subseteq \en{D}$ called a
          contingency for $t$, such that $t$ is a counterfactual cause
          for $\bar{a}$ in $D - \Gamma$.
	\end{itemize}
\end{definition}
\looseness -1
A tuple $t$ is a counterfactual cause, if by removing it from the
database, we remove $\bar{a}$ from the answer.  The tuple is an actual
cause if one can find a contingency under which it becomes a
counterfactual cause: more precisely, one has to find a set $\Gamma$
such that, after removing $\Gamma$ from the database we bring it to a
state where removing/inserting $t$ causes $\bar{a}$ to switch between
an answer and a non-answer.  Obviously, every counterfactual cause is
also an actual cause, by taking $\Gamma=\emptyset$.  The definition of
causality extends naturally to the case when the query $q$ is Boolean:
in that case, a counterfactual cause is a tuple that, when removed,
determines $q$ to become false.
	
\vspace{-1mm}
\begin{example}\label{ex:WhySo}
    \looseness -1
  Consider the query $q(x) \datarule R(x, y), S(y)$ on the following database instance, and assume all tuples are
  endogenous: $R=\en{R}$, $S=\en{S}$.  Consider the answer $a_2$.  The tuple
  $S(a_1)$ is a counterfactual cause for this result, because if we
  remove this tuple from $S$ then $a_2$ is no longer an answer.  Now
  consider the answer $a_4$.  Tuple $S(a_3)$ is not a counterfactual
  cause: if we remove it from $S$, $a_4$ is still an answer.  But
  $S(a_3)$ is an actual cause with contingency $\{S(a_2)\}$: once we
  remove $S(a_2)$ we reach a state where $a_4$ is still an answer, but
  further removing $S(a_3)$ makes $a_4$ a non-answer.

{
\centering
\vspace{-0mm}
\scriptsize
\mbox{
\hspace{8mm}
	\mbox{\begin{minipage}[t]{20mm}	
			\renewcommand\tabcolsep{7pt}
			\begin{tabular}[t]{|l|l|l}
			\multicolumn{2}{c}{R} \\
			\hline
			\textit{X} & \textit{Y}\\
			\hline
			$a_1$      &  $a_5$\\
			$a_2$      &  $a_1$\\
		        $a_3$      &  $a_3$ \\
			$a_4$      &  $a_3$  \\
                        $a_4$      &  $a_2$  \\
			\hline			\end{tabular}						
		\end{minipage}}\hspace{2mm}

\begin{minipage}[t]{100mm}
\setlength{\tabcolsep}{0.5mm}
	\mbox{\begin{minipage}[t]{15mm}
		\renewcommand\tabcolsep{7pt}
			\begin{tabular}[t]{|l|}
			\multicolumn{1}{c}{S} \\
			\hline
			\textit{Y}\\
			\hline
			$a_1$\\
			$a_2$\\
			$a_3$\\
			$a_4$\\
            $a_6$\\
			\hline
			\end{tabular}
	\end{minipage}}
	\mbox{\begin{minipage}[t]{25mm}
		\renewcommand\tabcolsep{7pt}
			\begin{tabular}[t]{ll|l|l}
			\multicolumn{4}{l}{$q(x) \datarule R(x, y)S(y)$} \\
			\cline{3-3}
			& & \textit{X} \\
			\cline{3-3}        
			& & $a_2$    \\        
			& & $a_3$    \\
			& & $a_4$    \\
			\cline{3-3}			
			\end{tabular}

		\end{minipage}}
\end{minipage}}
}

\looseness -1
  For a more subtle example, consider the Boolean query $q \datarule$
  $R(x,a_3),S(a_3)$ (where $a_3$ is a constant), which is true on the
  given instance.  Suppose only the first three
  tuples in $R$ are endogenous, and the last two are exogenous: $\ex{R} =
  \set{(a_4, a_3), (a_4, a_2)}$.  Let's examine whether $\en{R}(a_3,a_3)$
  is a cause for the query being true.  This tuple is not an actual
  cause.  This is because $\set{\en{S}(a_3)}$ is not a contingency for
  $\en{R}(a_3,a_3)$: by removing $\en{S}(a_3)$ from the database we make the
  query false, in other words the tuple $\en{R}(a_3,a_3)$ makes no
  difference, under any contingency.  Notice that
  $\set{\ex{R}(a_4,a_3)}$ is not contingency because
  $\ex{R}(a_4,a_3)$ is exogenous.
\end{example}
\vspace{-1mm}

\looseness -1
In this paper we discuss two instantiations of query causality.  In the
first, called {\whyso} causality, we are given an actual answer
$\bar a$ to the query, and would like to find the cause(s) for this
answer.  \Autoref{def:causality} is given for {\whyso} causality.
In this case $D$ is the real database, and the endogenous tuples $\en{D}$
are a given subset, while exogenous are $\ex{D} = D-\en{D}$.  In the second
instantiation, called {\whyno}  causality, we are given a
non-answer $\bar a$ to the query, i.e. would like to know the cause
why $\bar a$ is not an answer.  This requires some minor changes to
\Autoref{def:causality}.  Now the real database consists entirely of
exogenous tuples, $\ex{D}$.  In addition, we are given a set of
potentially missing tuples, whose absence from the database caused
$\bar a$ to be a non-answer: these form the endogenous tuples, $\en{D}$,
and we denote $D = \ex{D} \cup \en{D}$.  We do not discuss in this paper how
to compute $\en{D}$: this has been addressed in recent
work~\cite{DBLP:journals/pvldb/HuangCDN08}.
In this setting, the definition of the {\whyno}  causality is the
dual of \autoref{def:causality} and we give it here briefly: a {\em
  counterfactual cause} for the non-answer $\bar a$ in $\ex{D}$ is a
tuple $t\in \en{D}$ s.t. $\ex{D} \not\models q(\bar a)$ and $\ex{D} \cup
\set{t} \models q(\bar a)$; an actual cause for the non-answer $\bar
a$ is a tuple $t \in \en{D}$ s.t. there exists a set $\Gamma \subseteq
\en{D}$ called contingency set s.t. $t$ is a counterfactual cause for the
non-answer of $\bar a$ in $\ex{D} \cup \Gamma$.

We now define {\em responsibility}, measuring the degree of
causality.

\begin{definition}[Responsibility]\label{def:responsibility}\!\!
  Let $\bar{a}$ be an answer or non-answer to a query $q$, and let $t$
  be a cause (either {\whyso}, or {\whyno}  cause).  The {\em
    responsibility} of $t$ for the (non-)answer $\bar{a}$ is:
	$$
	\rho_t = \frac{1}{1+ \min_{\Gamma} |\Gamma|}
	$$
        where $\Gamma$ ranges over all contingency sets for $t$.
\end{definition}
 
\looseness -1
Thus, the responsibility is a function of the minimal number of tuples
that we need to remove from the real database $D$ (in the case of
{\whyso}), or that we need to add to the real database $\ex{D}$
(in the case of {\whyno}) before it becomes counterfactual.  The
tuple $t$ is a counterfactual cause iff $\rho_t = 1$, and it is an actual
cause iff $\rho_t > 0$. By convention, if $t$ is not a cause, $\rho_t = 0$.

\vspace{-1mm}
\begin{example}[IMDB continued]
  \Autoref{fig:musicals-graph} shows the lineage of the answer \genre{Musical}
  in \autoref{ex:imdb}.  Consider the movie \movie{Sweeney Todd}: its
  responsibility is $1/3$ because the smallest contingency is:
  {\{\sqll{Director}(David, Burton), \sqll{Director}(Humphrey, Burton)\}}
(if we remove both directors, then \movie{Sweeney Todd}
  becomes counterfactual).  Consider now the movie \movie{Manon Lescaut}:
  its responsibility is $1/5$ because the smallest contingency set is
  {\{\sqll{Director} \director{(David, Burton)},
  \sqll{Movie}(\movie{Flight}),
  \sqll{Movie}(\movie{Candide}), 
  \sqll{Director}(Tim, Burton)\}}.
\end{example}
\vspace{-1mm}

\looseness -1
We now define formally the problems studied in this paper.  Let $D =
\ex{D} \cup \en{D}$ be a database consisting of endogenous and exogenous
tuples, $q$ be a query, and $\bar{a}$ be a potential answer to the
query.

\begin{description}[itemsep=1pt, parsep=1pt, topsep = 1pt]
\item[Causality problem] Compute the set $C \subseteq \en{D}$ of
  actual causes for the answer $\bar{a}$.
\item[Responsibility problem] For each actual cause $t \in C$, compute
  its responsibility $\rho_t$.
\end{description}

\looseness -1
We study the {\em data complexity} in this paper: the query $q$ is
fixed, and the complexity is a function of the size of the database
instance $D$.  In the rest of the the paper we restrict our discussion
w.l.o.g.\ to Boolean queries: if $q(\bar x)$ is not Boolean, then to
compute the causes or responsibilities for an answer $\bar a$ it
suffices to compute the causes or responsibilities of the Boolean
query $q[\bar a/\bar x]$, where all head variables are substituted
with the constants in $\bar a$.

\section{Complexity of  Causality} 
\label{sec:complexity_analysis_causality}
\looseness -1
We start by proving that causality can be computed efficiently; even
stronger, we show that causes can be computed by a relational query.
This is in contrast with the general causality problem, where
Eiter~\cite{DBLP:journals/ai/EiterL02} has shown that deciding
causality for a Boolean expression is NP-complete.  We obtain
tractability by restricting our queries to conjunctive queries.
Chockler et al. \cite{DBLP:journals/tocl/ChocklerHK08} have shown that
causality for ``read once'' Boolean circuits is in PTIME.  Our results
are strictly stronger: for the case of conjunctive queries without
self-joins, queries with read-once lineage expressions are precisely
the hierarchical queries~\cite{pods:invited:2007,OlteanuH2009}, while
our results apply to all conjunctive queries.  The results in this
section apply uniformly to both {\whyso} and {\whyno} 
causality, so we will simply refer to causality without specifying
which kind.  Also, we restrict our discussion to Boolean queries only.

\looseness -1
We write positive Boolean expressions in DNF, like $\Phi = (X_1 \wedge
X_3) \vee (X_1 \wedge X_2 \wedge X_3) \vee (X_1 \wedge X_4)$;
sometimes we drop $\wedge$, and write $\Phi = X_1 X_3 \vee X_1 X_2 X_3
\vee X_1 X_4$.  A conjunct $c$ is {\em redundant} if there exists
another conjunct $c'$ that is a strict subset of $c$.  Redundant
conjuncts can be removed without affecting the Boolean expression.  In
our example, $X_1 X_2 X_3$ is redundant, because it strictly contains
$X_1 X_3$; it can be removed and $\Phi$ simplifies to $ X_1 X_3 \vee
X_1 X_4$.  A positive DNF is {\em satisfiable} if it has at least one
conjunct; otherwise it is equivalent to \sqll{false} and we call it
{\em unsatisfiable}.

\looseness -1
Next, we review the definition of lineage.  Fix a Boolean conjunctive
query consisting of $m$ atoms, $q = g_1, \ldots, g_m$, and database
instance $D$; recall that $D = \ex{D} \cup \en{D}$ (exogenous and
endogenous tuples). For every tuple $t \in D$, let $X_t$ denote a
distinct Boolean variable associated to that tuple.  A {\em valuation}
for $q$ is a mapping, $\theta : \Var(q) \rightarrow Adom(D)$, where
$Adom(D)$ the active domain of the database, such that the
instantiation of every atom is a tuple in the database:
$t_i=\theta(g_i)\in D$ for $i=1,\ldots,m$.  We associate to the
valuation $\theta$ the following conjunct: $c^{\theta} = X_{t_1}
\wedge \ldots \wedge X_{t_m}$.  The {\em lineage} of $q$ is:
  $$\Phi = \bigvee_{\theta : q \rightarrow D}  c^{\theta}$$

\looseness -1
We will assume w.l.o.g.\ that $\ex{D} \not\models q$ and $(\ex{D} \cup \en{D})
\models q$ (otherwise we have no causes).

\begin{definition}[$\enSymb$-lineage]  The {\em $\enSymb$-lineage} of $q$ is:
  \begin{align*}
    \en{\Phi} &  =  \Phi[X_{t}:=\sqll{true}, \forall t \in \ex{D}]    
  \end{align*}
\end{definition}
\looseness -1
Here $\Phi[X_{t}:=\sqll{true}, \forall t \in \ex{D}]$ means
substituting $X_t$ with $\sqll{true}$, for all Boolean variables
$X_t$ corresponding to exogenous tuples $t$.  Thus, the $\enSymb$ -lineage is
obtained as follows.  Compute the standard lineage, over all tuples
(exogenous and endogenous), then set to \sqll{true} all exogenous
tuples: the remaining expression depends only on endogenous tuples.
The following technical result allows us to
compute the causes to answers of conjunctive queries.

\begin{theorem}[Causality]\label{thm:causalityEq}\!\! 
  Let $q$ be a conjunctive query, and $t$ be an endogenous tuple.  Then the
  following three conditions are equivalent:
  \begin{enumerate}[itemsep=1pt, parsep=1pt, topsep = 1pt]
  \item $t$ is an actual cause for $q$ (\autoref{def:causality}).
  \item There exists set of tuples $\Gamma \subseteq \en{D}$ such that
    the lineage $\Phi[X_u = \sqll{false},$ $\forall u \in \Gamma]$ is
    satisfiable, and $\Phi[X_u = \sqll{false},$ $\forall u \in \Gamma;
    X_t = \sqll{false}]$ is unsatisfiable. 
  \item There exists a non-redundant conjunct in the $\enSymb$-lineage 
    $\en{\Phi}$ that contains the variable $X_t$.
  \end{enumerate}
\end{theorem}
\looseness -1
We give the proof in the Appendix.  The theorem gives a PTIME
algorithm for computing all causes of $q$: compute the $\enSymb$-lineage
$\en{\Phi}$ as described above, and remove all redundant conjuncts.  All
tuples that still occur in the lineage are actual causes of $q$.

\vspace{-1mm}
\begin{example} \label{ex:lineage}\!\!  
    \looseness -1
    Consider $q \datarule
  R(x,y),S(y),y=a_3$ over the database of \autoref{ex:WhySo}.  Its
  lineage is $\Phi = X_{R(a_3,a_3)} X_{S(a_3)} \vee X_{R(a_4,a_3)}$ $X_{S(a_3)}$.  Assume $R(a_4,a_3)$ is exogenous and $R(a_3,a_3)$,
  $S(a_3)$ are endogenous.  Then the $\enSymb$-lineage $\en{\Phi}$ is obtained by
  setting $X_{R(a_4,a_3)} = \sqll{true}$: $\en{\Phi} = X_{R(a_3,a_3)}
  X_{S(a_3)} \vee X_{S(a_3)}$.  After removing the redundant conjunct,
  the $\enSymb$-lineage becomes $\en{\Phi} = X_{S(a_3)}$; hence, $S(a_3)$ is the
  only actual cause for the query.
\end{example}
\vspace{-1mm}

\looseness -1
In the rest of this section we prove a stronger result.  Denote $C_R$
the set of actual causes in the relation $R$; that is, $C_R\subseteq
\en{R}$, and every tuple $t \in C_R$ is an actual cause.  We show that
$C_R$ can be computed by a relational query.  In particular, this
means that the causes to a (non-)answer can be computed by a SQL
query, and therefore can be performed entirely in the database system.

\begin{theorem}[Causality FO]\label{thm:causalQuery}~\!\!\!
  Given a Boolean query $q$ over relations $R_1,\ldots R_k$, the set
  of all causes of $q$ $\set{C_{R_1},\ldots,C_{R_k}}$ can be expressed
  in non-recursive stratified Datalog with negation, with only two
  strata.
\end{theorem}
\looseness -1
\Autoref{thm:causalQuery}, shows that causes can be expressed in a
language equivalent to a subset of first order logic~\cite{abiteboul-hull-vianu}
and that, moreover, only one level of negation is needed. The proof is
in the appendix.

\vspace{-1mm}
\begin{example} \label{ex:rs}\!\!\!  
    \looseness -1
    Continuing with the query $q :-
  R(x,y),S(y)$ from \autoref{ex:lineage}), suppose all tuples in $S$
  are endogenous. Thus, we have $\ex{R}, \en{R}, \en{S}$, but $\ex{S}
  = \emptyset$.  The complete Datalog program that produces the causes
  for $q$ is:
  \begin{align*}
    I(y) 	& \datarule   \ex{R}(x,y),\en{S}(y) \\
    C_R(x,y) & \datarule  \en{R}(x,y),\en{S}(y),\neg I(y) \\
    C_S(y) 	& \datarule  \en{R}(x,y),\en{S}(y), \neg I(y) \\
    C_S(y) 	& \datarule  \ex{R}(x,y),\en{S}(y)
   \end{align*}
   The role of $\neg I(y)$ is to remove redundant terms from the
   lineage.  To see this, consider the database $R = \set{(a_4,a_3),
     (a_3,a_3)}$, $S = \en{S} = \set{a_3}$, and assume that $\en{R} =
   \set{(a_3, a_3)}$, $\ex{R} = \set{(a_4, a_3)}$, thus, $q$'s lineage
   and $\enSymb$-lineage are:
   \begin{align*}
     \Phi & = X_{R(a_4,a_3)}X_{S(a_3)} \vee X_{R(a_3,a_3)}X_{S(a_3)} \\
     \en{\Phi} & = X_{S(a_3)} \vee X_{R(a_3,a_3)}X_{S(a_3)} \equiv X_{S(a_3)}
   \end{align*}
   Thus, the only actual cause of $q$ is $S(a_3)$.  Consider $C_R$,
   which computes causes in $R$.  Without the negated term $\neg
   I(y)$, $C_R$ would return $R(a_3,a_3)$ (which would be incorrect).
   The role of the negated term $\neg I(y)$ is to remove the redundant
   terms in $\en{\Phi}$: in our example, $C_R$ returns the empty set (which
   is correct).  Similarly, one can check that $C_S$ returns $S(a_3)$.
   Note that negation is necessary in $C_R$ because it is
   non-monotone: if we remove the tuple $R(a_4,a_3)$ from the database
   then $R(a_3,a_3)$ becomes a cause for the query $q$, thus $C_R$ is
   non-monotone.  Hence, in general, we must use negation in order to
   compute causes.
\end{example}
\vspace{-1mm}

\vspace{-1mm}
\begin{example} \label{ex:rsr}\!\!\!  
    \looseness -1
    Consider $q \datarule S(x), R(x,y),
  S(y)$, and assume that $S$ is endogenous and $R$ is exogenous: in
  other words, $S=\en{S}$, $R=\ex{R}$. The following Datalog program
  computes all causes:
  \begin{align*}
    I(x) 	& \datarule  \en{S}(x), \ex{R}(x,x) \\
    C_S(x) 	& \datarule  \en{S}(x),\ex{R}(x,y),\en{S}(y), \neg I(x), \neg I(y) \\
    C_S(y) 	& \datarule  \en{S}(x),\ex{R}(x,y),\en{S}(y), \neg I(x), \neg I(y)
   \end{align*}
   Here, too, we can prove that $C_S$ is non-monotone and, hence, must
   use negation.  Consider the database instance $R = \set{(a_4,a_3),$
$(a_3,a_3)}$, $S = \set{a_3, a_4}$.  Then $S(a_4)$ is not a cause;
   but if we remove $R(a_3, a_3)$, then $S(a_4)$ becomes a cause.
\end{example}
\vspace{-1mm}

\looseness -1
As the previous examples show, the causality query $C$ is, in general,
a non-monotone query: by inserting more tuples in the database, we
determine some tuples to no longer be causes.  Thus, negation is
necessary in order to express $C$.  The following corollary gives a
sufficient condition for the causality query to simplify to a
conjunctive query.

\begin{corollary}\!\!\!\label{cor:conj}
  Suppose that each relation $R_i$ is either endogenous or exogenous
  (that is, either $\en{R}_i = R_i$ or $\ex{R}_i = R_i$).  Further, suppose
  that, if $R_i$ is endogenous, then the relation symbol $R_i$ occurs
  at most once in the query $q$.  Then, for each relation name $R_i$,
  the causal query $C_{R_i}$ is a single conjunctive query (in
  particular it has no negation).
\end{corollary}

\looseness -1
The two examples above show that the corollary is tight:
\autoref{ex:rs} shows that causality is non-monotone when a relation
is mixed endogenous/exogenous, and \autoref{ex:rsr} shows that
causality is non-monotone when the query has self-joins, even if all
relations are either endogenous or exogenous.  

To illustrate the corollary, we revisit \autoref{ex:rs}, where the
query is $q \datarule R(x,y),S(y)$, and assume that $\ex{R} = \emptyset$ and
$\ex{S} = \emptyset$.  Then the Datalog program becomes:
  \begin{align*}
    C_R(x,y) 	& \datarule  \en{R}(x,y),\en{S}(y) \\
    C_S(y)  	& \datarule  \en{R}(x,y),\en{S}(y)
   \end{align*}

\section{Complexity of Responsibility} 
\label{sec:complexity_analysis_responsibility}
\looseness -1
In this section, we study the complexity of computing responsibility. As before, we restrict our discussion to Boolean
queries.  Thus, given a Boolean query $q$ and an endogenous tuple $t$,
compute its responsibility $\rho_t$ (\autoref{def:responsibility}).
We say that the query is in PTIME if there exists a PTIME algorithm
that, given a database $D$ and a tuple $t$ computes the value
$\rho_t$; we say that the query is NP-hard, or simply hard, if the
problem ``given a database instance $D$ and a number $v$, check
whether $\rho_t > v$'' is NP-hard.  The strongest result in this section and the paper is
a \emph{dichotomy theorem} for {\whyso} queries without self-joins: for
every query, computing the \emph{responsibility is either in PTIME or
NP-hard} (\autoref{sub:why_so_r}).  The case of non-answers ({\whyno}) turns out to be a
simpler problem as \autoref{sub:why_no_r} shows.

\subsection{Why So?} 
\label{sub:why_so_r}
\looseness -1
We assume that the conjunctive query $q$ is without
self-joins, i.e.\ every relation occurs at most once in $q$; we
discuss self-joins briefly at the end of the section.  W.l.o.g.\ we further
assume that each relation is either fully endogenous or exogenous
($\en{R}_i = R_i$ or $\ex{R}_i = R_i$).  Recall that computing the
{\whyso} responsibility of a tuple $t$ requires computing the
smallest contingency set $\Gamma$, such that $t$ is a counterfactual
cause in $D-\Gamma$.  We start by giving three hard
queries, which play an important role in the dichotomy result.

\begin{theorem}[Canonical Hard Queries]\label{thm:hardQueries}~\!\!
  Each of the following three queries is NP-hard:
  \begin{align*}\!\!
    h_1^* & \datarule \en{A}(x), \en{B}(y), \en{C}(z), W(x,y,z) \\
    h_2^* & \datarule \en{R}(x,y),\en{S}(y,z),\en{T}(z,x) \\
    h_3^* & \datarule \en{A}(x), \en{B}(y), \en{C}(z), R(x,y), S(y,z), T(z,x)
	\end{align*}
If the type of a relation is not specified, then the
query remains hard 
whether the relation is endogenous or exogenous.
\end{theorem}
\looseness -1
We give the proof in the Appendix: we prove the hardness of $h_1^*$ and
$h_2^*$ directly, and that of $h_3^*$ by using a particular reduction from
$h_2^*$.  Chockler and Halpern~\cite{DBLP:journals/jair/ChocklerH04}
have already shown that computing responsibility for Boolean circuits
is hard, in general.  One may interpret our theorem as strengthening
that result somewhat by providing three specific queries whose
responsibility is hard.  However, the theorem is much more significant.  
We show in this section that every query that is
hard can be proven to be hard by a simple reduction from one of these
three queries.

\looseness -1
Next, we illustrate PTIME queries, and start with a trivial example $q
\datarule R(a,y)$ where $a$ is a constant. If $t = R(a,b)$, then its
minimum contingency is simply the set of all tuples $R(a,c)$ with $c
\neq b$, and one can compute $t$'s responsibility by simply counting
these tuples.  Thus, $q$ is in PTIME.  We next give a much more subtle example.

\vspace{-1mm}
\begin{example}[Ptime Query]\label{ex:ptimeQuery}~\!\!\!\!
    \looseness -1
	  Let $q \datarule R(x,y),{S}(y,z)$, let both $R$ and $S$ be endogenous, and w.l.o.g.\ let $t$ be a tuple in
	  ${R}$.  We show how to compute the size of the minimal
	  contingency set $\Gamma$ for $t$ with a reduction to the
	  max-flow/min-cut problem in a network.  Given the database instance
	  $D$, construct the network illustrated in
	  \autoref{fig:figs_flowTransform}.  Its vertices are partitioned into
	  $V_1 \cup \ldots \cup V_5$.  $V_1$ contains the {\em source}, which
	  is connected to all nodes in $V_2$. There is one edge $(x,y)$ from
	  $V_2$ to $V_3$ for every tuple $(x,y) \in {R}$, and one edge
	  $(y,z)$ from $V_3$ to $V_4$ for every tuple $(y,z) \in {S}$.
	  Finally, every node in $V_4$ is connected to the {\em target}, in
	  $V_5$.  Set the capacity of all edges from the source or into the
	  target to $\infty$. The other capacities will be described shortly.
	  Recall that a cut in a network is a set of edges $\Gamma$ that
	  disconnect the source from the target. A \emph{min-cut} is a cut of minimum
	  capacity, and the capacity of a min-cut can be computed in PTIME
	  using Ford-Fulkerson's algorithm.  Now we make an important
	  observation: any mincut $\Gamma$ in the network corresponds to a set
	  of tuples\footnote{In other words, the mincut cannot include the
	    extra edges connected to the source or the target as they
	    have infinite capacity.}  in the database $D = {R} \cup
	  {S}$, such that $q$ is false on $D - \Gamma$.  We use this fact
	  to compute the responsibility of $t$ as follows:  First, set the
	  capacity of $t$ to 0, and that of all other tuples in ${R},
	  {S}$ to 1.  Then, repeat the following procedure for every path
	  $p$ from the source to the target that goes through $t$: set the
	  capacities of all edges\footnote{In our example, $p-\set{t}$
	    contains a single other edge (namely a tuple in ${S}$).  
	For longer queries, it may contain additional edges.
	    For the query
	    ${R}(x,y),{S}(y,z),{T}(z,u)$, for example, $p-\set{t}$ 
	    always contains two edges.  Hence we refer to edges in $p -
	    \set{t}$ in the plural.}  in $p-\set{t}$ to $\infty$, compute
	  the size of the mincut, and reset their capacities back to 1.
	  In \autoref{fig:figs_flowTransform} there are two such paths $p$:
	  the first is $x_1,y_2,z_1$ (the figure shows the capacities set for
	  this path), the other path is $x_1,y_2,z_2$.  We claim that for
	  every mincut $\Gamma$, the set $\Gamma-\set{t}$ is a contingency set
	  for $t$.  Indeed, $q$ is false on $D-\Gamma$ because the source is
	  disconnected from the target,  and $q$ is true on $(D-\Gamma) \cup
	  \set{t}$, because once we add $t$ back, it will join with the other
	  edges in $p-\set{t}$. Note that $\Gamma$ cannot include these edges as
	  their capacity is $\infty$.  Thus, by repeating for all paths $p$
	  (which are at most $|{S}|$), we can compute the
	  size of the minimal contingency set as $\min|\Gamma| - 1$.

	\begin{figure}[tb]
		\centering
			\includegraphics[scale=0.5]{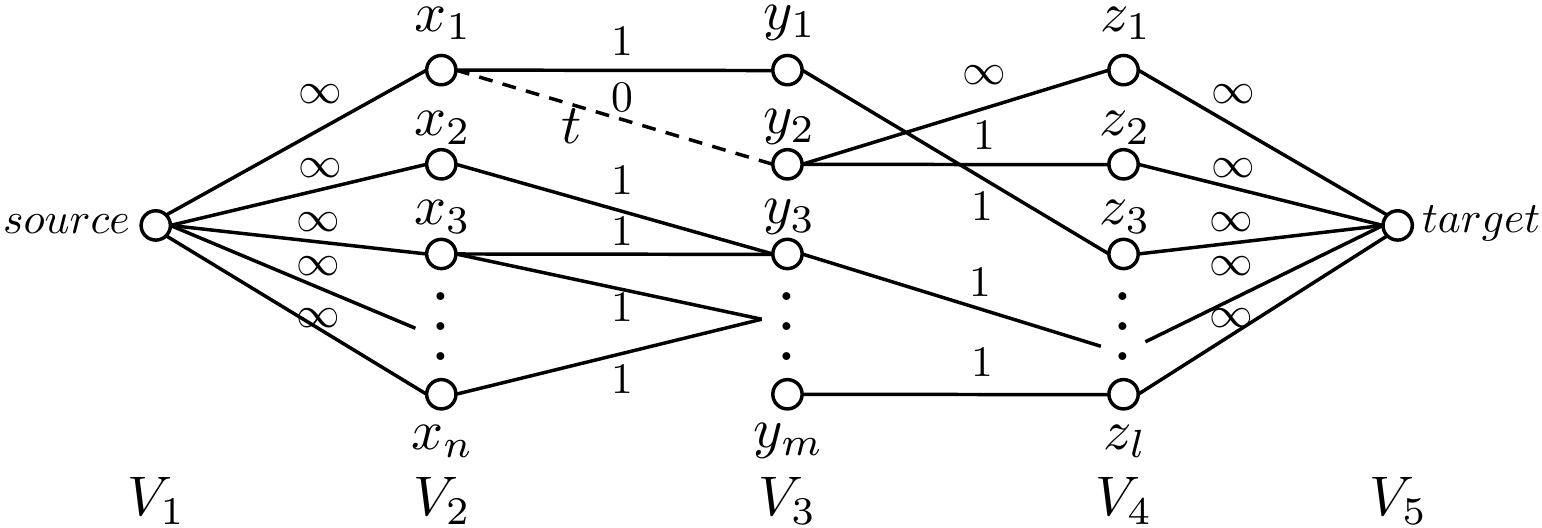}
			\vspace{-0.1in}
		\caption{Flow transformation for $q \datarule R(x,y),S(y,z)$.}
		\label{fig:figs_flowTransform}
		\vspace{-0.1in}
	\end{figure}
\end{example}
\vspace{-1mm}

We next generalize the algorithm in \autoref{ex:ptimeQuery} to the large class of
{\em linear queries}. 
We need two definitions first.

\begin{definition}[Dual Query Hypergraph $\mathcal{H}^{D}$]\label{def:dualHypergraph}~\!\!\!\!\!
  The \emph{dual query hypergraph}
  $\mathcal{H}^D(V,\mathcal{E})$ of a query $q \datarule g_1,\ldots, g_m$ is a
  hypergraph with vertex set $V=\{g_1,\ldots,g_m\}$ and a hyperedge
  $E_i$ for each variable $x_i\in \Var(q)$ such that $E_i=\{g_j\,|\,
  x_i\in\Var(g_j)\}$.
\end{definition}

\looseness -1
Note that nodes are the atoms, and edges are the variables.
This is the ``dual'' of the standard query
hypergraph~\cite{gottlob2001}, where nodes are variables and edges are
atoms.

\begin{figure}[tb]
   \centering
   
       \subfloat[$q$]
		{\includegraphics[scale=0.42]{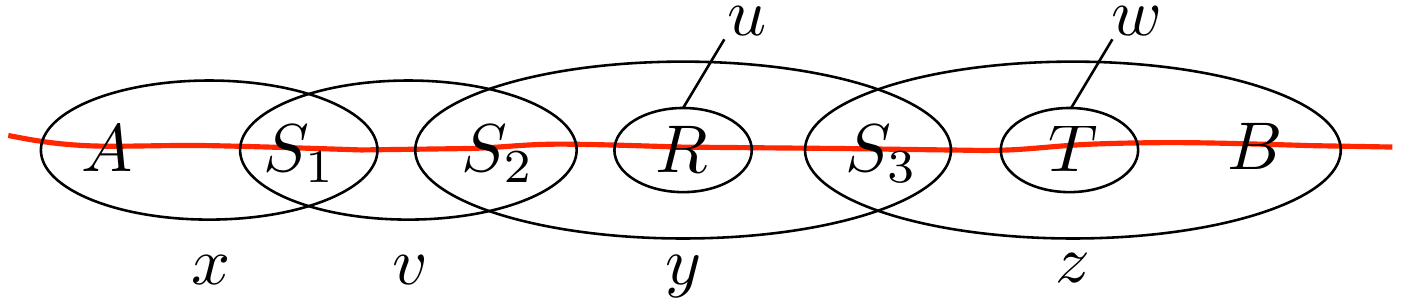}
		\label{fig:linearHypergraph}}
		\hspace{2mm}
		\subfloat[$h_1^*$]
           	{\includegraphics[scale=0.36]{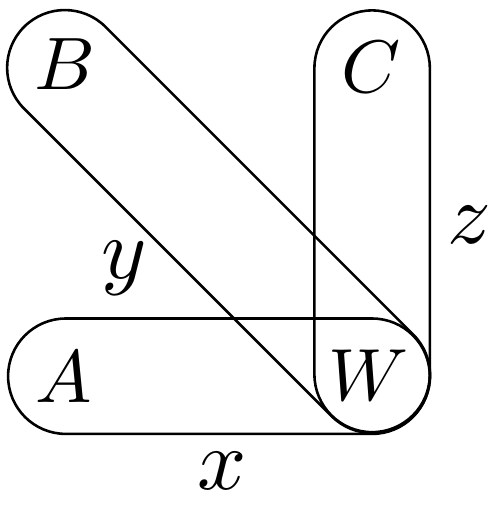}
    		\label{fig:hardQuery}}
    		\vspace{-0.1in}
              \caption{Dual query hypergraphs for easy query {\small
            ${q \datarule A(x), S_1(x,v), S_2(v,y), R(y,u), S_3(y,z),
              T(z,w), B(z)}$}, and hard query {\small
                    ${h_1^* \datarule A(x), B(y), C(z), W(x,y,z)}$}}
           \label{fig:example_hypergraphs}
           \vspace{-0.1in}
\end{figure}

\begin{definition}[Linear Query]\label{def:linearHypergraph}~\!\!\!\!\!
  A hypergraph $H(V,\mathcal{E})$ is linear if there exists
  a total order $S_V$ of $V$, such that every hyperedge
  $e\in\mathcal{E}$ is a consecutive subsequence of $S_V$.
	A  query is \emph{linear} if its dual hypergraph is linear.
\end{definition}
\looseness -1
In other words a query is linear if its atoms can be ordered such that
every variable appears in a continuous sequence of atoms.  For
example, the query $q$ in \autoref{fig:linearHypergraph} is
linear. Order the atoms as $A, S_1,$ $S_2,R,S_3,T,B$, and \mbox{every} variable
appears in a continuous sequence, e.g.\  $y$ occurs in $S_2,R,S_3$.
On the other hand, none of the queries in
\autoref{thm:hardQueries} is linear.  For example, the dual hypergraph
of $h_1^*$ is shown in \autoref{fig:hardQuery}: one cannot ``draw a
line'' through the vertices and stay inside hyperedges.  Note that
the definition of linearity ignores the endogenous/exogenous status of
the atoms.

\looseness -1
For every linear query, the responsibility of a tuple can be computed in PTIME
using \specificref{Algorithm}{alg:flowTransform}.  The algorithm
essentially extends the construction given \autoref{ex:ptimeQuery} to
arbitrary linear queries.  Note that it treats endogenous relations
differently than exogenous by assigning to them weight $\infty$.  Thus, we have:

\begin{theorem}[Linear Queries]\label{thm:linearPTIME}
  For any linear query $q$ and any endogenous tuple $t$, the
  responsibility of $t$ for $q$ can be computed in PTIME in the size
  of the database $D$.
\end{theorem}

\begin{algorithm2e}[t]
    {\small
\caption{Calculating responsibility for linear queries}\label{alg:flowTransform}
\SetKwInput{Function}{Function}
\SetKwFunction{flowGraph}{flowGraph}
\SetKwFunction{capacity}{capacity}
\SetKwFunction{maxFlow}{maxFlow}
\SetKwFor{ForAll}{forall}{do}{endfch}	
\SetFuncSty{sqll}

\Indm
\KwIn{$q \datarule g_1,\ldots,g_m$, $D$ and $t$, \KwOut{$\rho_t$}}	
\BlankLine
\Indp

$G=\flowGraph(\textit{dualHypergraph}(q),D)$ \;

\ForAll{source-target paths $p=\{e_1,\ldots,e_m\}\in G$, $t\in p$}
    {$\capacity(e_i) \leftarrow \infty$, $\capacity(t) \leftarrow  0$ \;
    $\Gamma_j \leftarrow  \maxFlow(G)$ \;}
\Return $\rho_t=\frac{1}{1+(min_j|\Gamma_j|-1)}$ \;

\BlankLine
\Indm
\Function{\FuncSty{\flowGraph}$(\mathcal{H},D)$}
\Indp
$L=\{g_1(\bar x_1),\ldots,g_m(\bar x_m)\}$ linearization of $\mathcal{H}$ \;
$V=\{\{\sqll{source}\},V_1',V_1,V_2',V_2\ldots,V_m',V_m,\{\sqll{target}\}\}$ \;
\ForAll{$g_i$}
{
    $V_i' \leftarrow \{t_j\,|\, t_j\in q(\bar x_i') \datarule g_{i-1},g_i\}$, $\bar x_i' \leftarrow \bar x_{i-1}\cap\bar x_i$ \;
    $\forall u\in V_{i-1}$, $v\in V_i'$, add edge $e(u,v)$ if $u(\bar x_i')=v(\bar x_i')$ \;
    $\capacity(e) \leftarrow \infty$ \;
    \ForAll{$t_j\in g_i$}
    {
        $V_i \leftarrow V_i'\cup\{t_j\}$ \;
        $\forall u\in V_i'$, $v\in V_i$, add edge $e(u,v)$ if $u(\bar x_{i}')=v(\bar x_i')$ \;
        \lIf{$t_j\in \en{D}$} 
        {$\capacity(e)\!\leftarrow\! 1$} 
        \lElse $\capacity(e)\!\leftarrow\!\infty$ \;
    }
}
$\forall v\in V_m$, $\capacity(e(v,\sqll{target})) \leftarrow \infty$ \;
$\forall v\in V_1'$, $\capacity(e(\sqll{source},v)) \leftarrow \infty$ \;
}
\end{algorithm2e}

\looseness -1
So far, \autoref{thm:hardQueries} has described some hard queries, and
\autoref{thm:linearPTIME} some PTIME queries.  Neither class is
complete, hence we do not yet have a dichotomy yet.  To close the gap
we need to work on both ends. We start by expanding the class of
hard queries.

\begin{definition}[rewriting $\rewrite$]\label{def:rewriting}~\!\!
    We define the following 
    rewriting relation on conjunctive queries without self-joins: $q$
  rewrites to $q'$, in notation $q \rewrite q'$, if $q'$ can be
  obtained from $q$ by applying one of the following three rules:
  \begin{itemize}[leftmargin=2.2\parindent, itemsep=1pt, parsep=1pt, topsep = 1pt]
  
    \item \textsc{Delete $x$ ($q \rewrite q[\emptyset/x]$):}
    Here, $q[\emptyset/x]$ denotes the query obtained by removing the
    variable $x \in Var(q)$, and thus decreasing the arity of all atoms that
    contained $x$.

  \item \textsc{Add $y$ ($q \rewrite q[(x,y)/x]$):} Here, $q[(x,y)/x]$ denotes
    the query obtained by adding variable $y$ to all atoms that
    contain variable $x$, and thus increasing their arity, provided there exists an
    atom in $q$ that contains both variables $x,y$.

  \item \textsc{Delete $g$ ($q \rewrite q - \set{g}$):} 
    Here, $g$ denotes an atom and $q - \set{g}$ denotes the query $q$
    without the atom $g$, provided that $g$ is exogenous, or there exists some other atom $g_0$ s.t. $\Var(g_0)\subseteq \Var(g)$. 
  \end{itemize}
\end{definition}

\looseness -1
Denote $\!\!\rewriteclosed\!\!$
the transitive and reflexive closure of
$\rewrite$.  We show that rewriting always reduces complexity:

\begin{lemma}[Rewriting]\label{lemma:hard}
  If $q \rewrite q'$ and $q'$ is NP-hard, then $q$ is
  also NP-hard.  In particular,  $q$ is NP-hard if $q \rewriteclosed h_i$,
  where $h_i$ is one of the three queries in
  \autoref{thm:hardQueries}.
\end{lemma}

\vspace{-1mm}
\begin{example}[Rewriting]
  We illustrate how one can prove that the query $q \datarule
  R(x,y),S(y,z),T(z,u),K(u,x)$ is hard, by rewriting it to $h_2$:
  \begin{align*}
    q & \rewrite 
    R(x,y),S(y,z),T(x,z,u),K(u,x) 				&\mbox{\ \ (add $x$)} \\
    & \rewrite  R(x,y),S(y,z),T(x,z,u),K(u,x,z) &\mbox{\ \ (add $z$)} \\
    & \rewrite  R(x,y),S(y,z),T(x,z,u) 			&\mbox{\ \ (delete $K$)} \\
    & \rewrite  R(x,y),S(y,z),T(x,z) 			&\mbox{\ \ (delete $u$)} 
  \end{align*}
\end{example}
\vspace{-1mm}

\looseness -1
With rewriting we expanded the class of hard queries. Next we expand the class of PTIME queries. As notation, we say that two atoms $g_i,
g_j$ of a conjunctive query are {\em neighbors} if they share a
variable: $\texttt{Var}(g_i)\cap \texttt{Var}(g_j)\neq\emptyset$.

\begin{definition}[Weakening $\weakening$] \label{def:weakening}
  We define the following {\em weakening} relation on conjunctive
  queries without self-joins: $q$ weakens to $q'$, in notation $q
  \weakening q'$, if $q'$ can be obtained from $q$ by applying
  one of the following two rules:
  \begin{itemize}[leftmargin=2.2\parindent, itemsep=1pt, parsep=1pt, topsep = 1pt]
  	\item \textsc{Dissociation} If $\ex{g}$ is an exogenous atom and $v_i$ a
    variable occurring in some of its neighbors, then let $q'$ be
    obtained by adding $v_i$ to the variable set of $\ex{g}$ (this increases its arity).
  	\item \textsc{Domination} If $\en{g}$ is an endogenous atom and there exists
    some other endogenous atom $\en{g}_0$ s.t.\ $\Var(\en{g}_0)
    \subseteq \Var(\en{g})$, then let $q'$ be obtained by making
    $g$ exogenous, $\ex{g}$.
  \end{itemize}
\end{definition}

Intuitively, a minimum contingency never needs to contain tuples from a dominated relation, and thus the relation is effectively exogenous.
Along the lines of \autoref{lemma:hard}, we show the following for weakening:

\begin{lemma}[Weakening]\label{lemma:weakening}
    If $q\weakening q'$ and $q'$ is in PTIME, then $q$ is also in PTIME.
\end{lemma}

Thus, weakening allows us to expand the class of PTIME queries.
We denote $\!\!\weakeningclosed\!\!$ the transitive and reflexive closure of
$\weakening$.  We say that a query $q$ is {\em weakly linear} if there
exists a weakening $q \weakeningclosed q'$ s.t. $q'$ is
linear. Obviously, every linear query is also weakly linear.

\begin{corollary}[Weakly Linear Queries] \label{lemma:weakly:linear}~\!\!\!\!
	If $q$ is weakly linear, then it is in PTIME.
\end{corollary}
\looseness -1
\Autoref{lemma:weakening} is based on the simple observation that a weakening produces a query $q'$ over a database instance $D'$ that produces the same output tuples as query $q$ on database instance $D$. Weakening only affects exogenous and dominated atoms, which are not part of minimum contingencies, and therefore responsibility remains unaffected.
This also implies an algorithm for computing responsibility in the case of weakly linear queries: find a weakening of $q$ that is linear and apply \specificref{Algorithm}{alg:flowTransform}.

\vspace{-1mm}
\begin{example}\label{ex:hard2ptime}  We illustrate the lemma with two
  examples.  First, we show that $q \datarule
  \en{R}(x,y),\ex{S}(y,z),\en{T}(z,x)$ is in PTIME by weakening with a
  dissociation:
  \begin{align*}
    q & \weakening \en{R}(x,y),\ex{S}(x,y,z),\en{T}(z,x) 	&\mbox{\ \ (dissociation)}
  \end{align*}
  The latter is linear. Query $q$ should be contrasted with
  $h_2^*$ in \autoref{thm:hardQueries}: the only difference is that here
  $\ex{S}$ is exogenous, and this causes $q$ to be in PTIME while
  $h_2^*$ is NP-hard.
  Second, consider $q \datarule
  \en{R}(x,y),\en{S}(y,z),\en{T}(z,x),\en{V}(x)$.  Here we weaken with a
  domination followed by a dissociation:
  \begin{align*}
    q & \weakening  \ex{R}(x,y),\en{S}(y,z),\ex{T}(z,x),\en{V}(x) 		&\mbox{ (domination)}\\
      & \weakening  \ex{R}(x,y,z),\en{S}(y,z),\ex{T}(z,x,y),\en{V}(x)	&\mbox{ (dissociation)}
  \end{align*}
  The latter is linear with the linear order
  $\en{S},\ex{R},\ex{T},\en{V}$.
\end{example}
\vspace{-1mm}

\looseness -1
We say that a query $q$ is {\em final} if it is not weakly linear and
for every rewriting $q \rewrite q'$, the rewritten query $q'$ is
weakly linear.  For example, each of $h_1^*, h_2^*, h_3^*$ in
\autoref{thm:hardQueries} is final: one can check that if we try to
apply any rewriting to, say, $h_1^*$ we obtain a linear query.  We can
now state our main technical result:

\begin{theorem}[Final Queries]\label{thm:final} 
	If $q$ is final, then $q$ is one of $h_1^*,h_2^*,h_3^*$.
\end{theorem}

This is by far the hardest technical result in this paper.  We give the proof in the appendix.  Here, we show how to use it to prove the
dichotomy result.

\begin{corollary}[Responsibility Dichotomy]
  Let $q$ be any conjunctive query without self joins.  Then:
  \begin{itemize}[leftmargin=2.2\parindent, itemsep=1pt, parsep=1pt, topsep = 1pt]
  \item If $q$ is weakly linear then $q$ is in PTIME.
  \item If $q$ is not weakly linear then it is NP-hard.
  \end{itemize}
\end{corollary}

\begin{proof}
    \looseness -1
  If $q$ is weakly linear then it is in PTIME by
  \autoref{lemma:weakly:linear}.  Suppose $q$ is not weakly linear.
  Consider any sequence of rewritings $q=q_0 \rewrite q_1
  \rewrite q_2 \rewrite \ldots$ Any such sequence must terminate
  as any rewriting results in a simpler query.  We rewrite as
  long as $q_i$ is not weakly linear and stop at the last query $q_k$
  that is not weakly linear.  That means that any further rewriting
  $q_k \rewrite q'$ results in a weakly linear query $q'$. In other
  words, $q_k$ is a final query.  By \autoref{thm:final}, $q_k$ is one
  of $h_1^*, h_2^*, h_3^*$.  Thus, we have proven $q
  \rewriteclosed h_j$, for some $j= 1,2,3$.  By
  \autoref{lemma:hard}, the query $q$ is NP-hard.
\end{proof}

\looseness -1
\introparagraph{Extensions} We have shown in \autoref{sec:complexity_analysis_causality} that causality can be computed with a relational query. This
raises the question: if the responsibility of a query $q$ is in PTIME,
can we somehow compute it in SQL?  We answer this negatively:

\begin{theorem}[Logspace]\label{thm:logspaceComp}~\!\!\!
  Computing the {\whyso} responsibility of a tuple $t\in \en{D}$ is hard for \textsc{logspace} 
  for the following query:
$q \datarule \en{R}(x,u_1,y),$ $\en{S}(y,u_2,z),\en{T}(z,u_3,w)$ 
\end{theorem}

\looseness -1
Finally, we add a brief discussion of queries with self-joins.  Here
we establish the following result:

\begin{proposition}[self-joins]\label{prop:selfJoin}
  Computing the responsibility of a tuple $t$ for $q \datarule
  \en{R}(x), \ex{S}(x,y), \en{R}(y)$ is NP-hard.  The same holds if
  one replaces $\ex{S}$ with $\en{S}$.
\end{proposition}
\looseness -1
We include the proof in the appendix.  Beyond this result, however,
queries with self-joins are harder to analyze, and we do not yet have
a full dichotomy.  In particular, we leave open the complexity of the
query $\en{R}(x,y),\en{R}(y,z)$.

\subsection{Why No?} 
\label{sub:why_no_r}
\looseness -1
While the complexity of {\whyso} responsibility turned out to be
quite difficult to analyze, the {\whyno}  responsibility is much
easier.  This is because, for any query $q$ with $m$ subgoals and
non-answer $\bar a$, any contingency set for a tuple $t$ will have at
most $m-1$ tuples.  Since $m$ is independent on the size of the
database, we obtain the following:

\begin{theorem}[{\whyno}  responsibility]\label{thm:whyNoResp}~\!\!\!\!\!
  Given a query $q$ over a database instance $D$ and a non-answer
  $\bar{a}$, computing the responsibility of $t\in \en{D}$ over $\en{D}$ is
  in PTIME.
\end{theorem}

\section{Related Work} 
\label{sec:related_work}
\looseness -1
Our work is mainly related to and unifies ideas from work on \emph{causality}, \emph{provenance}, and  \emph{query result explanations}.

\looseness -1
\introparagraph{Causality}
Causality is an active research area mainly in logic and philosophy with its own dedicated workshops (e.g.\ \cite{CausalityWorkshop:2009}). 
The idea of \emph{counterfactual causality} (if $X$ had not occurred, $Y$ would not have occurred) can be traced back to Hume \cite{Menzies:Causation2008}, and
the best known counterfactual analysis of causation in modern times is due to  Lewis \cite{Lewis1973}. 
Halpern and Pearl \cite{HalpernPearl:Cause2005} define a variation they call \emph{actual causality} 
which relies upon a graph structure called a causal network, and adds the crucial concept of a \emph{permissive contingency} before determining causality.
Chockler and Halpern \cite{DBLP:journals/jair/ChocklerH04} define the degree of \emph{responsibility} as a gradual way to assign causality.
Our definitions of {\whyso} and {\whyno}  causality and responsibility for conjunctive queries build upon the HP definition, but simplify it and do not require a causal network. A general overview of causality in a database context is given in \cite{DEBulletin2010}, while \cite{MUD2010} introduces functional causality as an improved, more robust version of the HP definition.

\looseness -1
\introparagraph{Provenance}
Approaches for defining data provenance can be mainly classified into three categories: 
how-, \mbox{why-,} and where-provenance 
\cite{DBLP:conf/icdt/BunemanKT01, DBLP:journals/ftdb/CheneyCT09, DBLP:journals/tods/CuiWW00, DBLP:conf/pods/GreenKT07}.
We point to the close connection between \emph{why-provenance} and {\whyso} causality: both definitions concern the same tuples if all tuples in a database are endogenous\footnote{\footnotesize{Note that why-provenance (also called minimal witness basis) defines a set of sets. To compare it with {\whyso} causality, we consider the union of tuples across those sets.}}.
However, our work extends the notion of provenance by 
allowing users to partition the lineage tuples into endogenous and exogenous, and presenting a strategy for constructing a query to compute all causes\footnote{Note that, in general, {\whyso} tuples are \emph{not} identical to the subset of endogenous tuples in the why-provenance.}.
In addition, we can rank tuples according to their individual responsibilities, and determine a gradual contribution with counterfactual tuples ranked first.

\looseness -1
\introparagraph{Missing query results}
Very recent work focuses on the problem of explaining missing query answers, 
i.e.\ why a certain tuple is not in the result set?
The work by Huang et al.~\cite{DBLP:journals/pvldb/HuangCDN08} and
the Artemis~\cite{DBLP:journals/pvldb/HerschelHT09} system
present provenance for potential answers by providing
\emph{tuple insertions} or modifications that would yield the missing
tuples.  This is equivalent to providing the set of endogenous tuples for {\whyno}  causality.
Alternatively, Chapman and Jagadish \cite{DBLP:conf/sigmod/ChapmanJ09} focus
on the \emph{operator} in the query plan that
eliminated a specific tuple,
and Tran and Chan~\cite{DBLP:conf/sigmod/TranC10} suggest an approach to automatically generate a modified query whose result includes both the original query's results as well as the missing tuple.

\looseness -1
Our definitions of {\whyso} and {\whyno}  causality highlight the symmetry between the two types of  provenance (``positive and negative provenance"). Instead of considering them in separate manners, we show how to construct Datalog programs that compute all {\whyso} or {\whyno}  tuple causes given a partitioning of tuples into endogenous and exogenous. Analogously, responsibility applies to both cases in a uniform manner.

\section{Conclusions} 
\label{sec:conclusions}
\looseness -1
In this paper, we introduce causality as a framework for explaining
answers and non-answers in a database setting. We define two kinds of
causality, {\whyso} for actual answers, and {\whyno} for
non-answers, which are related to the provenance of answers and
non-answers, respecitively. We demonstrate how to retrieve all causes for an answer
or non-answer using a relational query.  We give a comprehensive
complexity analysis of computing causes and their
responsibilities for conjunctive queries: whereas causality is shown to be always in PTIME, we present a dichotomy for responsibility within queries without self-joins. 

\introparagraph{Acknowledgements} We would like to thank Benny Kimelfeld for identifying a bug in an earlier version of one of the proofs.

\bibliography{\bibpath}
% \newpage

\appendix\phantomsection\label{sec:appendix}

\section{Nomenclature}\label{sec:nomenclature}

\begin{table}[h!]
\centering
\small
\begin{tabularx}{\linewidth}{  @{\hspace{0pt}} >{$}l<{$}  @{\hspace{2mm}} X @{}} 	
    \hline
    D                   & database instance \\
    A,B,C,R,S,T,W       & relations \\  
    \en{R}              & $R$ is a fully endogenous relation \\
    \ex{R}              & $R$ is a fully exogenous relation \\
    \en{D}              & set of endogenous tuples: $\en{D} \subseteq D$ for {\whyso}   \\
    \ex{D}              & set of exogenous tuples: $\ex{D} = D - \en{D}$    \\  
    \en{R}_i            & set of endogenous tuples in relation $R_i$    \\
    \Gamma              & contingency: $\Gamma \subseteq \en{D}$ \\
    C_{R_i}             & set of causes in relation $R_i$   \\
    \rho_t              & responsibility of tuple $t$   \\
    X_t                 & Boolean variable associated with tuple $t$    \\
    \Var(q)             & variables appearing in query $q$\\
    \Adom(D)            & active domain     \\
    \sg(x)              & set of subgoals containing variable $x$ \\
    \Phi                & lineage   \\
    \en{\Phi}           & endogenous lineage ($\enSymb$-lineage) \\
    \mathcal{H}^D       & dual hypergraph   \\
    \rewrite            & rewriting of a query \\
    \weakening          & weakening of a query \\
    h_1^*, h_2^*, h_3^* & canonical hard queries of \autoref{thm:hardQueries} \\
    \hline
    \end{tabularx}
\end{table}

\section{Proofs Causality} 
\label{appendix:proofsCausality}

\begin{proof}[\autoref{thm:causalityEq}]
	
	Assume $\Phi^{q(\bar a)}$ the lineage of $\bar a$ in $D$. We construct the endogenous lineage $\Psi^{q(\bar a)}=\Phi^{q(\bar a)}(\bar X_{\ex{D}}=\sqll{true})$, and $\Psi'$ a DNF with all the non-redundant clauses of $\Psi$. We will show that a variable $X_t$ is a cause of $\bar a$, iff $X_t\in \Psi'$, which means that $X_t$ is part of a non-redundant clause in the endogenous lineage of $\bar a$.
	
	\underline{Case A:} ({\whyso}, answer $\bar a$): 
	First of all, if $X_t$ is not in $\Psi'$, then $X_t$ is not a cause of $\bar a$, as there is no assignment that makes $X_t$ counterfactual for $\Psi'$ (and therefore $\Psi$), because of monotonicity.	
	If $X_t\in C_i$, where $C_i$ a clause of $\Psi'$, we select $\Gamma=\{ X_j\,|\, X_j\in\Psi'\, and\, X_j\notin C_i\}$ ($\Gamma\subseteq \en{D}$ since $\forall X_j\in\Psi'$, $X_j\in \en{D}$). Then, if we write $\Psi'=C_i\vee\Psi''$, we know that $\Psi''(X_\Gamma=\sqll{false})=\sqll{false}$, because $\Psi'$ contains only non-redundant terms. That means that every clause $C_j$ has at least one variable that is not in $C_i$, and therefore can be negated by the above choice of $\Gamma$.
	This makes $X_t$ counterfactual for $\Psi'$ (and also $\Psi$) with contingency $\Gamma$. Since $\Gamma\cap \ex{D}=\emptyset$ $X_t$ is also counterfactual for $\Phi^{q(\bar a)}$ with contingency $\Gamma$, meaning that $\Phi^{q(\bar a)}(X_\Gamma = \sqll{false})$ is satisfiable, and $\Phi^{q(\bar a)}(X_\Gamma = \sqll{false}, X_t = \sqll{false})$ is unsatisfiable.
	Therefore, conditions 1, 2, and 3 are equivalent.

	\underline{Case B:} ({\whyno}, non-answer $\bar a$):
	First of all, if $X_t$ is not in $\Psi'$, then $X_t$ is not a cause of $\bar a$, as there is no assignment that makes $X_t$ counterfactual for $\Psi'$ (and therefore $\Psi$), because of monotonicity.
	If $X_t\in C_i$, where $C_i$ a clause of $\Psi'$, we select $\Gamma'=\{ X_j\,|\, X_j\in C_i\, and\, X_j\neq X_t\}$, and assign $\Gamma = \en{D}-\Gamma'\cup\set{t}$. $\Gamma,\Gamma'\subseteq \en{D}$ since $\forall X_j\in\Psi'$, $X_j\in \en{D}$. Then, if we write $\Psi'=C_i\vee\Psi''$, we know that $\Psi''(X_{\Gamma}=\sqll{false})=\sqll{false}$, because $\Psi'$ contains only non-redundant terms. That means that every clause $C_j$ has at least one variable that is not in $C_i$, and therefore can be negated by the above choice of $\Gamma$.
	This makes $X_t$ counterfactual for $\Psi'$ (and therefore $\Psi$) with contingency $\Gamma'$. Since $\Gamma'\cap \ex{D}=\emptyset$ $X_t$ is also counterfactual for $\Phi^{q(\bar a)}$ with contingency $\Gamma'$, meaning that $\Phi^{q(\bar a)}(X_{\Gamma}= \sqll{false})$ is unsatisfiable, and $\Phi^{q(\bar a)}(X_{\Gamma} = \sqll{false}, X_t = \sqll{false})$ is satisfiable.
	Therefore, conditions 1, 2, and 3 are equivalent. 
\end{proof}

\begin{proof}[(\autoref{thm:causalQuery})]
To describe the relational query we need a number of technical
definitions.  Recall that $\en{R}_i, \ex{R}_i$ denote the
endogenous/exogenous tuples in $R_i$.  Given a Boolean conjunctive
query $q \datarule g_1(\bar x_1), \ldots, g_m(\bar x_m)$ we define a {\em
  refinement} to be a query of the form $r  \datarule g_1^{\varepsilon_1}(\bar x_1), \ldots,
  g_m^{\varepsilon_m}(\bar x_m) \label{eq:refinement}$,
where each $\varepsilon_i \in \set{\enSymb, \exSymb}$.  Thus, every atom is made
either exogenous or endogenous, and we call it an $\enSymb$- or an $\exSymb$-atom;
there are $2^m$ refinements.  Clearly, $q$ is logically equivalent to
the union of the $2^m$ refinements, and its lineage is equivalent to
the disjunction of the lineages of all refinements.  Consider any
refinement $r$.  We call a variable $x \in Var(r)$ an {\em $\enSymb$-variable}
if it occurs in at least one $\enSymb$-atom.  We apply repeatedly the
following two operations: (1) choose two $\enSymb$-variables $x,y$ and
substitute $y:=x$; (2) choose any $\enSymb$-variable $x$ and any constant $a$
occurring in the query and substitute $x:=a$.  We call any query $s$
that can be obtained by applying these operations any number of times
an {\em image} query; in particular, the refinement $r$ itself is a
trivial image.  There are strictly less than $2^{k^2}$ images, where
$k$ is the total number of $\enSymb$-variables and constants in the query. Note that $k$ is bounded by query size and thus irrelevant to data complexity. We
always minimize an image query.

Fix a refinement $r$.  We define an {\em $\enSymb$-embedding} for $r$ as a
function $e : r \rightarrow s$ that maps a strict subset of the
$\enSymb$-atoms in $r$ onto all $\enSymb$-atoms of $s$, where $s$ is the image of a
possibly different refinement $r'$.  Intuitively, an $\enSymb$-embedding is a
proof that a valuation for $r$ results in a redundant conjunct, because
it is strictly embedded a the conjunct of a valuation of $r'$.

We now describe in non-recursive, stratified Datalog the relational query that computes all causes
\begin{align}
  I_{s,e}(e(e^{-1}(\bar y))) & \datarule \sqll{atoms}(s)  \label{eq:i} \\
  C_{R_i}(\bar x_j) 		& \datarule  \sqll{atoms}(r) \wedge \bigwedge_{e : r \rightarrow s}
  \neg I_s(e^{-1}(\bar y)) \label{eq:c} 
\end{align}

There is one IDB predicate $I_{s,e}(e(e^{-1}(\bar y)))$, for each
possible embedding $e : r \rightarrow s$, and it appears in one single
rule, whose left hand side the same as $s$.  The head variables are
all $\enSymb$-variables $\bar y$ in $s$, where each $y$ is repeated
$|e^{-1}(y)|$ times: this is the purpose of the $e \circ e^{-1}$
function.  For example, if the embedding is $e : \en{R}(x_1,x_2,x_3)
\rightarrow \en{R}(y,y,y)$, then $e^{-1}(y) = (x_1,x_2,x_3)$ and
$e(e^{-1}(y)) = (y,y,y)$, hence the IDB is $I(y,y,y)$.  Next, there is
one IDB predicate $C_{R_i}$ for every relation name $R_i$, and there
are one or more rules for $C_{R_i}$: one rule \autoref{eq:c} for each
refinement $r$ and each $\enSymb$-atom $\en{\textit{g}}_j(\bar x_j)$ in $r$ that refers to
the relation $\en{R}_i$.

Therefore,  the Datalog program consisting of \autoref{eq:i} and \autoref{eq:c}
  computes the set of all causes to the Boolean query $q$ and returns
  them in the IDB predicates $C_1, \ldots, C_k$.
\end{proof}

\begin{proof}[(\autoref{cor:conj})]
	The proof is immediate: there exists a single refinement, which has no
	embedding.
\end{proof}

\section{Canonical Hard Queries} 
\label{appendix:canonicalQueries}

\begin{proof}[(\autoref{thm:hardQueries} $h_1^* \datarule R(x),S(y),T(z),W(x,y,z)$)]~\!\!
	We demonstrate hardness of $q_1$ with a reduction from the minimum vertex cover problem in a $3$-partite $3$-uniform hypergraph: Given an $3$-partite, $3$-uniform hypergraph and constant $K$, determine if there exists a vertex cover of size less or equal to $K$. This problem is shown to be hard in \cite{Senellart:2008p382}.
		
	Take a 3-partite 3-uniform hypergraph such as the one from \autoref{fig:3partiteHypergraph}. The nodes can be divided into 3 partitions ($R$, $S$ and $T$), such that every edge contains exactly one node from each partition. We construct 4 database relations $R(x)$, $S(y)$, $T(z)$ and $W(x,y,z)$. For each node in the $R$ partition of the hypergraph, we add a tuple in $R(x)$, and equivalently for $S$ and $T$. Also, for each edge of the hypergraph, we add a tuple in $W(x,y,z)$. Finally, we add an additional tuple to each relation: $r_0=(x_0)$, $s_0=(y_0)$, $t_0=(z_0)$ and $w_0=(x_0,y_0,z_0)$.
	
	\begin{figure}[tb]
		\centering
		\subfloat[]
		{\begin{minipage}{0.2\textwidth}
		\includegraphics[scale=0.48]{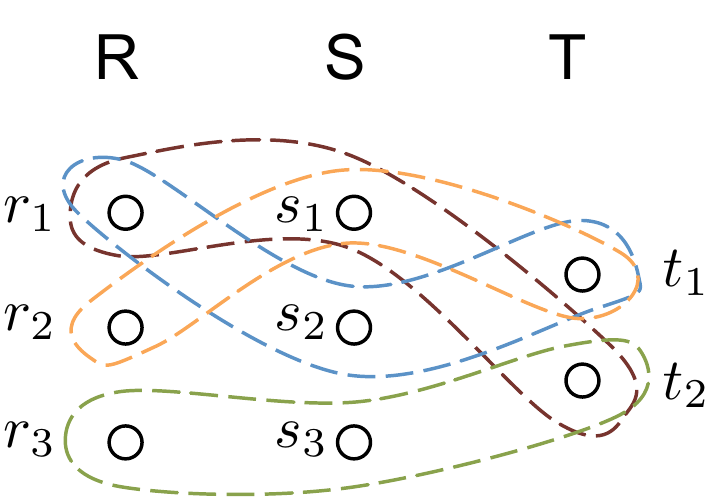}	
		\label{fig:3partiteHypergraph}
		\end{minipage}
		}
		\qquad
		\subfloat[]
		{\parbox{0.2\textwidth}{
		\renewcommand\tabcolsep{3pt}
		\begin{tabular}{|l|l|l|l|l|l|l|l|l|}
			\multicolumn{1}{l}{R} & \multicolumn{1}{l}{ } & \multicolumn{1}{l}{S} 
				& \multicolumn{1}{l}{ } & \multicolumn{1}{l}{T} & \multicolumn{1}{l}{ } & \multicolumn{3}{l}{W} \\
			\cline{1-1}\cline{3-3}\cline{5-5}\cline{7-9}
			\textit{X} & & $Y$ & & $Z$ & & $X$ & $Y$ & $Z$\\
			\cline{1-1}\cline{3-3}\cline{5-5}\cline{7-9}                                                
			$x_0$ &   & $y_0$ &   &  $z_0$ &  & $x_0$ & $y_0$ & $z_0$ \\
			$x_1$ &   & $y_1$ &   &  $z_1$ &  & $x_1$ & $y_1$ & $z_2$ \\
			$x_2$ &   &  $y_2$ &   &  $z_2$ &  & $x_1$ & $y_2$ & $z_1$ \\
			\cline{5-5}
			$x_3$ &   &  $y_3$ &  \multicolumn{3}{c|}{ }  & $x_2$ & $y_1$ & $z_1$ \\
			\cline{1-1}\cline{3-3}
			\multicolumn{6}{c|}{ }  & $x_3$ & $y_3$ & $z_2$ \\
			\cline{7-9}
		\end{tabular}\label{tbl:dbReduction}}
		}
		\caption{(a) Example 3-partite 3-uniform hypergraph. (b) Database instance created from the hypergraph of (a).}
	\end{figure}

	The database instance corresponding to the hypergraph of \autoref{fig:3partiteHypergraph} is shown in \autoref{tbl:dbReduction}. Now consider the join query 
	$$Q(x,y,z) \datarule R(x),S(y),T(z),W(x,y,z)$$
	The responsibility of tuple $r_0$, (equivalently $s_0$, $t_0$ or $w_0$), is equal to $\frac{1}{1+|\mathcal{S}|}$, where $\mathcal{S}$ is the minimum contingency set for tuple $r_0$. Therefore, $\mathcal{S}$ contains the minimum number of tuples that make $r_0$ counterfactual. Note that $s_0$, $t_0$ and $w_0$ cannot be contained in $\mathcal{S}$, as they are the only tuples that join with $r_0$.
	If $S$ a minimum contingency, then if $w_i\in \mathcal{S}$, then $\exists \mathcal{S}'=\{ \mathcal{S}\setminus\{ w_i\}\}\cup\{ r_j\}$, where $r_j.x=w_i.x$, and $\mathcal{S}'$ is also a contingency of the same size as $S$ and therefore minimum. Therefore, there exists minimum contingency $\mathcal{S}$ that only contains tuples from relations R, S and T. The tuples of R, S and T correspond to hypergraph nodes, and a contingency corresponds to a cover: if an edge was not covered, then there would exist a corresponding tuple that was not eliminated. Also the cover is minimum: if there existed a smaller cover, then there would exist a smaller contingency. Therefore, computing responsibility for $h_1^*$ is hard.
\end{proof}

\begin{proof}[(\autoref{thm:hardQueries} $h_2^* \datarule R(x,y),S(y,z),T(z,x)$)]~\!
We will demonstrate hardness for computing responsibility for $h_2^*$ using a reduction from 3SAT. 
For the reduction we will construct a 3-colored graph $G_\phi$. The nodes are colored with $a$, $b$, or $c$, and we call them $a$-nodes, $b$-nodes, and $c$-nodes respectively.  Any such graph corresponds uniquely to an $R,S,T$ instance: Every unique $a$, $b$, and $c$-node maps to a unique value from the domain of attribute $x$, $y$, and $z$ respectively. Therefore, $R$ contains all $a\rightarrow b$ edges, i.e. it contains all tuples $(u,v)$ where $u$ is an $a$-node and $v$ is a $b$-node and there is an edge $u\rightarrow v$, $S$ contains all $b\rightarrow c$ edges, and $T$ all $c\rightarrow a$ edges. The important property is that $h_2^*$ is true on the instance $R,S,T$ iff $G_\phi$ has a cycle of length 3; note that the nodes on such a cycle necessarily have colors $a$, $b$, $c$. From here on, we will use the term cycle to refer to a cycle of length 3.

Given a 3-colored graph $G_\phi$, we define the minimum contingency of $G_\phi$ as the smallest set of edges that contains at least one edge from each cycle. This notion is in fact directly equivalent to the notion of contingency $\Gamma$ for tuples defined in \autoref{def:responsibility}: a set $\Gamma$ is a minimum contingency for tuple $R(a_0,b_0)$ in $G_\phi\cup\{R(a_0,b_0), S(b_0,c_0), T(c_0,a_0)\}$, iff $\Gamma$ is a subset of the edges of $G_\phi$, and is a minimum contingency for $G_\phi$. We will show that given a 3-colored graph $G$ and a number $m$, checking if $G$ has a contingency set of size $\le m$ is NP-hard.

\vspace{2mm}
\underline{\sc Reduction of 3SAT to the colored graph $G_\phi$:}\\
Let $\phi=\bigwedge_{i=1}^nC_i$ be an instance of 3SAT, where $C_i$ are 3-clauses over some set $\bf X$ $=\{X_1,X_2,\ldots X_n\}$ of variables. For each variable $X_i$ we construct a graph $G_i$ called the \emph{local ring} of $X_i$, an example of which is given in \autoref{fig:figs_localRing}. The construction is as follows:
\begin{itemize}
    \item Pick $m_i$ odd and multiple of 3, such that $m_i \ge 9|C_{X_i}|$, where $|C_{X_i}|$ is the number of clauses the variable $X_i$ appears in. $m_i$ is the length of the ring $G_i$.
    \item Create two ordered sets of nodes $V^+=\{v_1^+,v_2^+,\ldots,v_{m_i}^+\}$ and $V^-=\{v_1^-,v_2^-,\ldots,v_{m_i}^-\}$.
    \item We assign the colors $a$, $b$, and $c$ to these nodes as follows: $V^+=\{a_1^+,b_2^+,c_3^+,a_4^+,\ldots\}$, and $V^-=\{a_1^-,b_2^-,c_3^-,a_4^-,\ldots\}$. In other words, every node $v_{3j+1}^+$ is an $a$-node (denoted as a square node in \autoref{fig:figs_localRing}), every node $v_{3j+2}^+$ is a $b$-node (denoted as a circle), and every node $v_{3j+3}^+$ is a $c$-node (denoted as a triangle). Similarly for $V^-$.
    \item Create the \emph{forward} edges $(v_j^+,v_{j+1}^-)$ and $(v_j^-,v_{j+1}^+)$ for all $j<m_i$, as well as  $(v_{m_i}^+,v_{1}^-)$ and $(v_{m_i}^-,v_{1}^+)$. Forward edges are shown as solid in \autoref{fig:figs_localRing}.
    \item Create the \emph{backward} edges $(v_j^+,v_{j-2}^+)$ and $(v_j^-,v_{j-2}^-)$ for all $j>2$, as well as $(v_{m_i-1}^+,v_{1}^+)$, $(v_{m_i-1}^-,v_{1}^-)$, $(v_{m_i}^+,v_{2}^+)$, and $(v_{m_i}^-,v_{2}^-)$. Backward edges are shown as dotted in \autoref{fig:figs_localRing}.
\end{itemize}
The collection of the local rings $G_i$ is the \emph{global graph} $G_\phi$. We will prove a series of lemmas that we need to demonstrate the hardness of computing the minimum contingency of $G_\phi$
\begin{figure}[tb]
    \centering
        \includegraphics[scale = 0.43]{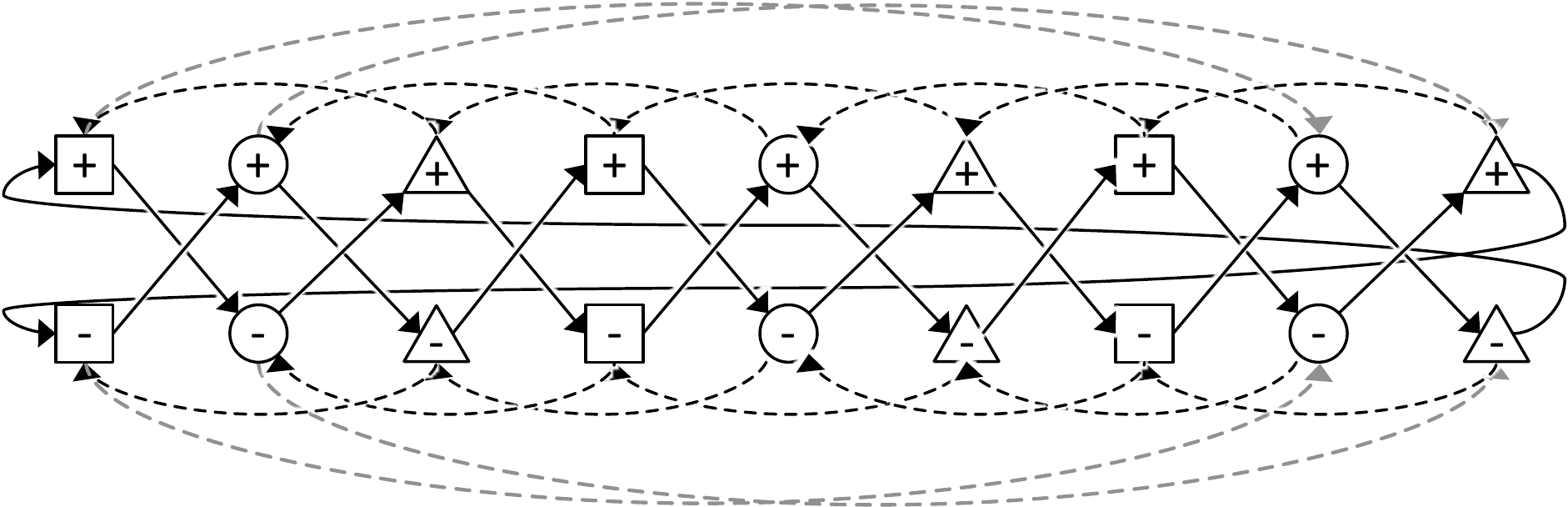}
    \caption{A local ring of length $m=9$.}
    \label{fig:figs_localRing}
\end{figure}

\begin{lemma}\label{lem:fwd_edge}
    A local ring $G_i$ has a minimum contingency comprised solely of forward edges.
\end{lemma}
\begin{proof}
    This is straightforward: all cycles in $G_i$ are comprised by 1 backward and 2 forward edges, and each backward edge participates in exactly one cycle. If a backward edge is part of a contingency set, it can always be replaced by one of the forward edges on the same cycle, resulting in a contingency set of at most the same size.
\end{proof}
From here on whenever we refer to a minimum contingency in a local ring, it will be implied that it includes forward edges only.

\begin{lemma}\label{lem:contingency_size}
    If $m_i$ is odd, there are exactly 2 minimum contingencies of size $m_i$ that include only forward edges.
\end{lemma}
\begin{proof}
    Since $m_i$ is odd, all forward edges in a local ring form a cyclic path going through all the nodes in $V^+\cup V^-$: 
    $$(a_1^+,b_2^-), (b_2^-,c_3^+), \ldots, (c_m^+,a_1^-), (a_1^-,b_2^+), \ldots,(c_m^-,a_1^+)$$

    Every two consecutive edges belong to a cycle (with the backwards edge), hence a contingency set has to contain one of the two edges. Thus, there are two minimum contingency sets: one consisting of all $(v^+,v^-)$ edges, the other consisting of all $(v^-,v^+)$ edges. 
\end{proof}
We will refer to these two contingency sets as $S_i^+$ and $S_i^-$ respectively, and they are associated with assignments $X_i=\sqll{true}$ and $X_i=\sqll{false}$ respectively.

We have chosen $m_i$ large enough to ensure that each clause containing $X_i$ can be associated with a unique sequence of 9 consecutive pairs of nodes in $G_i$: the first clause in $\phi$ that contains $X_i$ corresponds to nodes $\{v_1^+,\ldots, v_9^+,v_1^-,\ldots, v_9^-\}$ and the edges between them; the second clause in $\phi$ that contains $X_i$ corresponds to nodes $\{v_{10}^+,\ldots, v_{18}^+,v_{10}^-,\ldots, v_{18}^-\}$, and so on. Let $j$ be the starting position of $C$ in the local ring of $X_i$, and suppose that $X_i$ is part of the $k$th literal in $C$ ($k={1,2,\text{ or }3}$). If $L_k = X_i$, then $L_k$ corresponds to edge $(v_{j+k-1}^+,v_{j+k}^-)$ in the ring, else if $L_k = \neg X_i$, then $L_k$ corresponds to edge $(v_{j+k-1}^-,v_{j+k}^+)$. \Autoref{fig:figs_3SATclause} shows the edges corresponding to the literals of clause $C=X\vee Y\vee\neg Z$, assuming that $j$ is the starting position of clause $C$ in all three variable rings. 
Since we have set aside 9 pairs of nodes for this clause, in each of the 3 rings for the corresponding variables, we can always map the first literal to an $(a,b)$ edge ($R$ tuple), the second and third literals to a $(b,c)$ and $(c,a)$ edge respectively ($S$ and $T$ tuples).
\begin{figure}[tb]
    \centering
        \includegraphics[scale = 0.43]{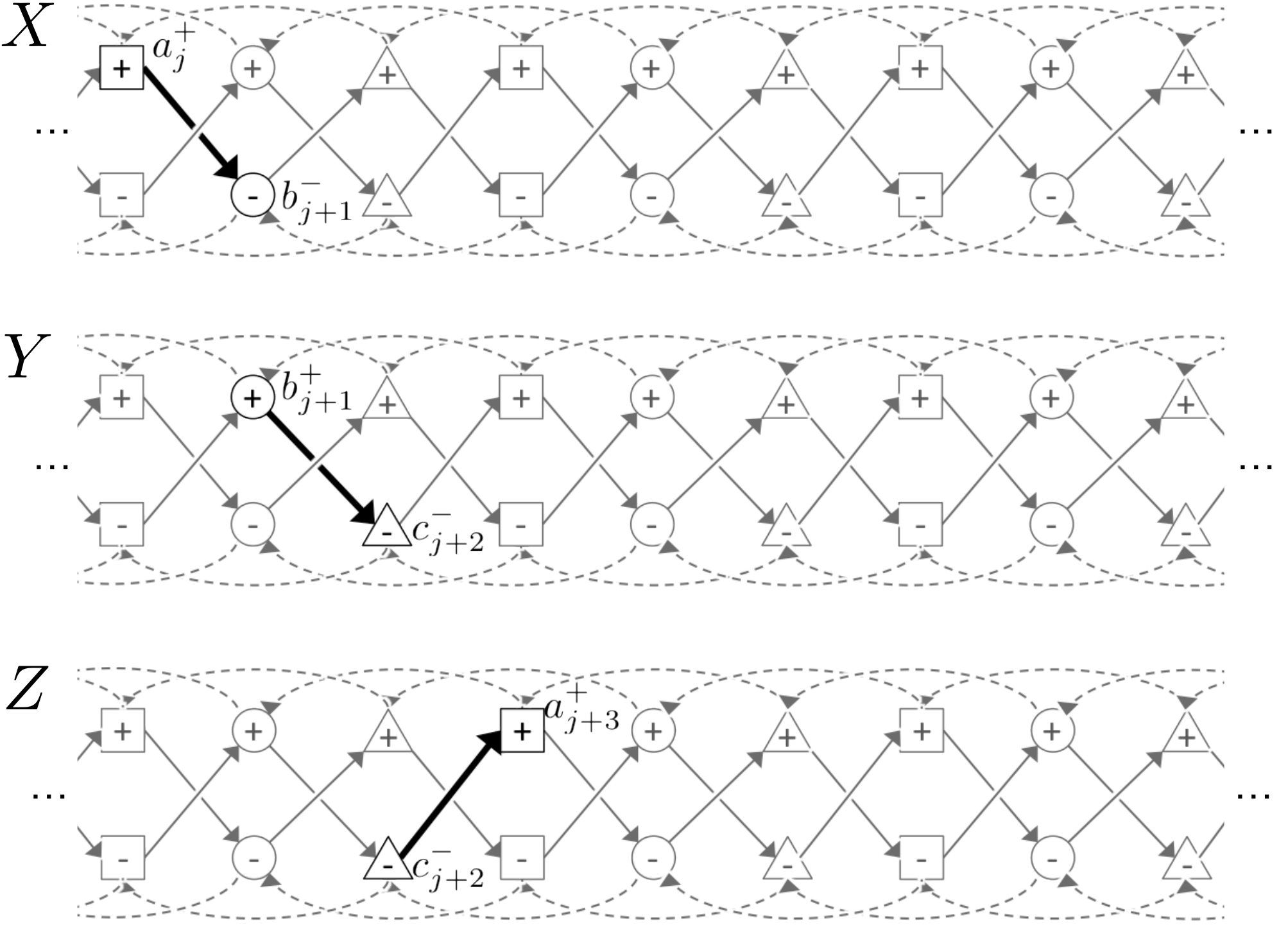}
    \caption{Depiction of clause $C=X\vee Y\vee \neg Z$. If $j$ is the starting position of $C$'s portion in the local rings of all 3 variables, then the literals of $C$ map to the 3 bold edges in the graph}
    \label{fig:figs_3SATclause}
\end{figure}

Let $(a_1,b_1)$, $(b_2,c_2)$, and $(c_3,a_3)$ be the edges corresponding to the 3 literals of a clause $C$ in rings $G_1$, $G_2$ and $G_3$ respectively. Then equate the following pairs of nodes: $a_1\equiv a_3$, $b_1\equiv b_2$, and $c_2\equiv c_3$. In the example of \autoref{fig:figs_3SATclause}, $X.a_j^+\equiv Z.a_{j+3}^+$, $X.b_{j+1}^-\equiv Y.b_{j+1}^+$ and $Y.b_{j+2}^-\equiv Z.c_{j+2}^-$, where $X_i.u$ denotes node $u$ in the ring of $X_i$.  

Equating two nodes simply means collapsing them into a single node, causing the local rings to intersect. Therefore, each clause corresponds to a new cycle in $G_\phi$ comprised of three forward edges from three variable rings. 
In the example of \autoref{fig:figs_3SATclause}, the cycle caused by $C$ is $X.a_j^+ \rightarrow X.b_{j+1}^-\equiv Y.b_{j+1}^+ \rightarrow Y.c_{j+2}^-\equiv Z.c_{j+2}^-\rightarrow Z.a_{j+3}^+ \equiv X.a_j^+$.
No additional cycles can be created due to the ``buffer'' nodes between clause portions in the local rings: a second clause on $X_i$ will map to edges that are distanced by 5 or more nodes from the previous clause edges. Note that one buffer node is actually sufficient and thus the local rings can be made smaller. We chose to have portions of length 9 for each clause to simplify some notation, as now every clause starts with an $a$-node, but that is not actually necessary.

The global graph $G_\phi$ therefore contains all cycles from the local rings $G_i$, plus one cycle for each clause in $\phi$. We will now show that $\phi$ has a satisfying assignment iff $G_\phi$ has a contingency of size $\sum_i m_i$, i.e. equal to the sum of the lengths of all local rings $G_i$.

\begin{lemma}
    The formula $\phi$ has a satisfying assignment iff $G_\phi$ has a contingency of size $\sum_i m_i$.
\end{lemma}
\begin{proof}
    From \autoref{lem:contingency_size} we know that a contingency for $G_i$ cannot be smaller than $m_i$, and in fact will either be the set $S_i^+$, or the set $S_i^-$ (defined after \autoref{lem:contingency_size}).  Any contingency for $G_\phi$ has to eliminate the cycles in $G_i$ as well, and therefore has to contain $S_i^+$ or $S_i^-$.

    We will first show the forward direction. Assume that $\phi$ has a satisfying assignment $A$. Construct set $\mathcal{S}$ by selecting $S_i^+$ or $S_i^-$ for each variable $X_i$ as follows: if $X_i=\sqll{true}$ in $A$, then add the edges from $S_i^+$ to $\mathcal{S}$, otherwise add the edges from $S_i^-$. Then $\mathcal{S}$ has size $\sum_i m_i$. Assume clause $C$, and the cycle of forward edges due to its literals. At least one literal $L$ of $C$ evaluates to \sqll{true} under assignment $A$, since $A$ is a satisfying assignment for $\phi$. If $L=X_i$, then $L$ maps to some edge $e=(v^+,v^-)\in S_i^+$, and since $X_i=\sqll{true}$ in $A$, $e\in\mathcal{S}$. Similarly, if $L=\neg X_i$, then $L$ maps to some edge $e=(v^-,v^+)\in S_i^-$, but again $e\in\mathcal{S}$ since $X_i$ evaluates to \sqll{false}. Therefore, every cycle due to clauses of $\phi$ has an edge in $\mathcal{S}$, and so $\mathcal{S}$ is a contingency for $G_\phi$.

    We will now show the reverse. Let $\mathcal{S}$ be a contingency in $G_\phi$, such that $|\mathcal{S}|=\sum_i m_i$. A contingency in $G_\phi$ also defines a contingency in all $G_i$ of respective sizes $m_i$; $\mathcal{S}$ contains either $S_i^+$ or $S_i^-$, which map to $X_i=\sqll{true}$ and $X_i=\sqll{false}$ respectively. Let $A$ be the assignment based on these values for each $X_i$, and let $C$ be a clause of $\phi$. Since $\mathcal{S}$ is a contingency in $G_\phi$, then at least one edge corresponding to a literal of $C$ is contained in $\mathcal{S}$. Let that literal be $L$ over variable $X_i$. If $L=X_i$, then the edge contained in $\mathcal{S}$ is $e=(v^+,v^-)$, and thus $S_i^+\subset\mathcal{S}$. This means that $X_i=\sqll{true}$ in $A$, so $C$ is satisfied. Similarly, if $L=\neg X_i$, then the edge contained in $\mathcal{S}$ is $e=(v^-,v^+)$, and thus $S_i^-\subset\mathcal{S}$. This means that $X_i=\sqll{false}$ in $A$, and again $C$ is satisfied. Therefore, assignment $A$ satisfies all clauses meaning that $\phi$ is satisfiable. 
\end{proof}
% 
% 
% \underline{\sc Mapping to a data instance:}\\
% We will now show how to create a data instance based on the graph $G_\phi$. Let every unique $a$, $b$ and $c$-node in $G_\phi$ be a different value from the domain of attributes $x$, $y$ and $z$ respectively. Create a tuple $R(a_i,b_j)$ for every edge $(a_i,b_j)$, a tuple $S(b_i,c_j)$ for every edge $(b_i,c_j)$, and a tuple $T(c_i,a_j)$ for every edge $(c_i,a_j)$. Name this dataset $D$. Let  $D'=D\cup\{R(a_0,b_0), S(b_0,c_0), T(c_0,a_0)\}$. Computing responsibility for tuple $R(a_0,b_0)$ in $D'$ corresponds to finding the minimal set of tuples in $D$ whose removal eliminates all join results. But the tuples in $D$ have an 1-to-1 mapping to edges in $G_\phi$, and a join result of $h_2^*$ over $D$ corresponds to a 3-length cycle in $G_\phi$. Therefore, the minimum contingency for $R(a_0,b_0)$ in $D'$ is equal to the minimum global contingency in $G_\phi$, which we've shown to be hard. 
% Therefore, computing responsibility for $h_2^*$ is NP-hard.

\end{proof}

%end deprecate

\begin{proof}[(\autoref{thm:hardQueries} $h_3^*$)]
We prove hardness of $h_3^*$ by a simple reduction from $h_2^*$. We start by writing the two queries as
\begin{align*}
	h_2^* \datarule& R(x,y), S(y,z), T(z,x)\\
	h_3^* \datarule& R'(x',y'), S'(y',z'), T'(z',x'), A'(x'), B'(y'), C'(z')
\end{align*}

Now we transform a database instance for $h_2^*$ into one of $h_3^*$ as follows:
For every tuple $r_i$ in $R(x,y)$, insert $r_i$ as new tuple into $A(x')$. Repeat analogously for each $s_i, t_i$ from $S(y,z), T(z,x)$ and $B'(x'), C'(z')$. Then, for each valuation $\theta = [a/x, b/y, c/z]$ that makes $h_2$ true over $D$, insert $(r_i, s_i)$, $(s_i, t_i)$, and $(t_i, r_i)$ into $R'$, $S'$, and $T'$, respectively, where $(r_i, s_i, t_i)$ represent the tuples in the original $R$, $S$, and $T$ corresponding to $\theta$. Now we have a one-to-one correspondence between a tuples in $R$ and $A'$, $S$ and $B'$, and $T$ and $C'$. Comparing the lineages for these two queries, one sees that $R'$, $S'$, and $T'$ are dominated by $A'$, $B'$, and $C'$, and that the minimal lineages are identical. Hence causes and their responsibility are identical.
\end{proof}

\begin{figure}[t]
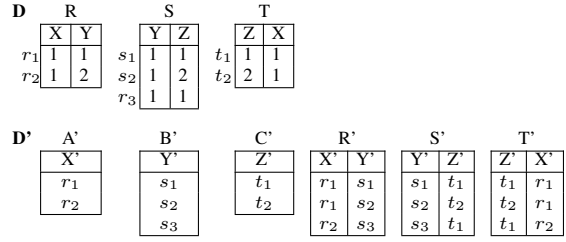

	{\scriptsize
	\centering\renewcommand\tabcolsep{3pt}
	\begin{tabular}{r|l|l|l|l|l|l|l|l|l|l|l|l|l|l|l|l|l|l|l|l|l|l|l|}
	\multicolumn{1}{l}{\textbf{D}} 
		& \multicolumn{2}{c}{R} & \multicolumn{1}{l}{ } 
		& \multicolumn{2}{c}{S} & \multicolumn{1}{l}{ } 
		& \multicolumn{2}{c}{T} & \multicolumn{1}{l}{ } \\
	\cline{2-3}\cline{5-6}\cline{8-9}
	\multicolumn{1}{l}{ } 
		& \multicolumn{1}{|l}{X} & \multicolumn{1}{|l|}{Y} & \multicolumn{1}{l}{ } 
		& \multicolumn{1}{|l}{Y} & \multicolumn{1}{|l|}{Z} & \multicolumn{1}{l}{ } 
		& \multicolumn{1}{|l}{Z} & \multicolumn{1}{|l|}{X} & \multicolumn{1}{l}{ }  \\				
	\cline{2-3}\cline{5-6}\cline{8-9}
	\multicolumn{1}{l}{ } 
		$r_1$\!\!	& \multicolumn{1}{|l}{1} & \multicolumn{1}{|l|}{1} & \multicolumn{1}{l}{ } 
		$s_1$\!		& \multicolumn{1}{|l}{1} & \multicolumn{1}{|l|}{1} & \multicolumn{1}{l}{ } 
		$t_1$\! 	& \multicolumn{1}{|l}{1} & \multicolumn{1}{|l|}{1} & \multicolumn{1}{l}{ }  \\				
	\multicolumn{1}{l}{ } 
		$r_2$\!\!	& \multicolumn{1}{|l}{1} & \multicolumn{1}{|l|}{2} & \multicolumn{1}{l}{ } 
		$s_2$\!		& \multicolumn{1}{|l}{1} & \multicolumn{1}{|l|}{2} & \multicolumn{1}{l}{ } 
		$t_2$\!		& \multicolumn{1}{|l}{2} & \multicolumn{1}{|l|}{1} & \multicolumn{1}{l}{ }  \\				
	\cline{2-3}\cline{8-9}
	\multicolumn{1}{l}{ } 
 		& \multicolumn{3}{l}{ } 
		$r_3$\!	& \multicolumn{1}{|l}{1} & \multicolumn{1}{|l|}{1} & \multicolumn{1}{l}{ } 
		& \multicolumn{3}{l}{ }  \\	
	\cline{5-6}					
	\multicolumn{1}{l}{} \\
	\multicolumn{1}{l}{\textbf{D'}} & \multicolumn{2}{c}{A'} & \multicolumn{1}{l}{ } 
		& \multicolumn{2}{c}{B'} & \multicolumn{1}{l}{ } 
		& \multicolumn{2}{c}{C'} & \multicolumn{1}{l}{ } 
		& \multicolumn{2}{c}{R'} & \multicolumn{1}{l}{ } 
		& \multicolumn{2}{c}{S'} & \multicolumn{1}{l}{ } 
		& \multicolumn{2}{c}{T'} & \multicolumn{1}{l}{ } \\
	\cline{2-3}\cline{5-6}\cline{8-9}\cline{11-12} \cline{14-15} \cline{17-18}
	\multicolumn{1}{l}{ } &  \multicolumn{2}{|c|}{X'} & \multicolumn{1}{l}{ } 
		& \multicolumn{2}{|c|}{Y'} & \multicolumn{1}{l}{ } 
		& \multicolumn{2}{|c|}{Z'} & \multicolumn{1}{l}{ } 
		& \multicolumn{1}{|l}{X'} & \multicolumn{1}{|l|}{Y'} & \multicolumn{1}{l}{ } 
		& \multicolumn{1}{|l}{Y'} & \multicolumn{1}{|l|}{Z'} & \multicolumn{1}{l}{ } 
		& \multicolumn{1}{|l}{Z'} & \multicolumn{1}{|l|}{X'} & \multicolumn{1}{l}{ }  \\				
	\cline{2-3}\cline{5-6}\cline{8-9}\cline{11-12} \cline{14-15} \cline{17-18}
	\multicolumn{1}{l}{ } &  \multicolumn{2}{|c|}{$r_1$} & \multicolumn{1}{l}{ } 
		& \multicolumn{2}{|c|}{$s_1$} & \multicolumn{1}{l}{ } 	
		& \multicolumn{2}{|c|}{$t_1$} & \multicolumn{1}{l}{ } 	
		& \multicolumn{1}{|l}{$r_1$} & \multicolumn{1}{|l|}{$s_1$} & \multicolumn{1}{l}{ } 
		& \multicolumn{1}{|l}{$s_1$} & \multicolumn{1}{|l|}{$t_1$} & \multicolumn{1}{l}{ } 
		& \multicolumn{1}{|l}{$t_1$} & \multicolumn{1}{|l|}{$r_1$} & \multicolumn{1}{l}{ }  \\				
	\multicolumn{1}{l}{ } &  \multicolumn{2}{|c|}{$r_2$} & \multicolumn{1}{l}{ } 
		& \multicolumn{2}{|c|}{$s_2$} & \multicolumn{1}{l}{ } 	
		& \multicolumn{2}{|c|}{$t_2$} & \multicolumn{1}{l}{ } 	
		& \multicolumn{1}{|l}{$r_1$} & \multicolumn{1}{|l|}{$s_2$} & \multicolumn{1}{l}{ } 
		& \multicolumn{1}{|l}{$s_2$} & \multicolumn{1}{|l|}{$t_2$} & \multicolumn{1}{l}{ } 
		& \multicolumn{1}{|l}{$t_2$} & \multicolumn{1}{|l|}{$r_1$} & \multicolumn{1}{l}{ }  \\				
	\cline{2-3}\cline{8-9}
	\multicolumn{1}{l}{ } &  \multicolumn{3}{l}{ } 
		& \multicolumn{2}{|c|}{$s_3$} & \multicolumn{1}{l}{ } 	
		& \multicolumn{3}{l}{ } 
		& \multicolumn{1}{|l}{$r_2$} & \multicolumn{1}{|l|}{$s_3$} & \multicolumn{1}{l}{ } 
		& \multicolumn{1}{|l}{$s_3$} & \multicolumn{1}{|l|}{$t_1$} & \multicolumn{1}{l}{ } 
		& \multicolumn{1}{|l}{$t_1$} & \multicolumn{1}{|l|}{$r_2$} & \multicolumn{1}{l}{ }  \\				
	\cline{5-6}\cline{11-12} \cline{14-15} \cline{17-18}
	\end{tabular}}
\caption{Example instance $D$ for query $h_2^*$ and corresponding instance $D'$ for query $h_3^*$.}
\end{figure}

\section{Responsibility Dichotomy} 
\label{appendix:responsibilityDichotomy}

\begin{proof}[(\autoref{thm:linearPTIME})]
    It is straightforward from \specificref{Algorithm}{alg:flowTransform}. The flow graph constructed has one edge per database tuple. The capacities of exogenous tuples are $\infty$ and all other tuples have capacity one. Every unit of s-t flow corresponds to an output tuple of $q$. It is impossible that a flow does not correspond to a valid output tuple, as the partitions are ordered based on the linearization. This means that if a variable is chosen it will not occur again later in the flow, and therefore we can't have invalid flows going through one value of $x$ in one partition and another in a later one. The steps of the algorithm: construction of the hypergraph, linearization, flow transformation, and computation of the maximum flow are all in PTIME and therefore responsibility of linear queries can be computed in PTIME.
\end{proof}

\begin{proof}[(\autoref{lemma:hard})]
    \underline{Case 1:} $q'$ resulted from variable deletion ($q \rewrite q[\emptyset/x]$). Then we can polynomially reduce $q'$ to $q$ by setting variable $x$ to the constant $a$. Therefore, if $q$ is in PTIME, $q'$ is also in PTIME.
    
    \underline{Case 2:} $q'$ resulted from rewrite ($q \rewrite q[(x,y)/x]$). Assume $q\datarule g_1(x,\ldots)g_2(x,y,\ldots)q_0$, then $q'\datarule g_1'(x,y,\ldots)g_2'(x,y,\ldots)q_0'$.
    Reduce $q'$ to $q$ as follows: For each subgoal $g_1'\in q'$, create a unique value $(x_i,y_j)$ for each tuple $g_1'(x_i,y_j)$, and assign these as the tuples of $g_1(x)$ in $q$. Similarly, create a unique value $(x_i,y_j)$ for each tuple $g_2'(x_i,y_j)$, and assign each as a tuple $g_2((x_i,y_j),y_j)$ to form relation $g_2$. Both queries have the same output and the same contingencies. Therefore, if $q$ is easy, $q'$ also has to be easy.
    
    \underline{Case 3:} $q'$ resulted from atom deletion ($q \rewrite q - \set{g}$). 
    Assume $q\datarule g_1,\ldots,g_m$ and $q'\datarule g_1,\ldots g_{j-1},g_{j+1},\ldots,g_m$. $q'$ differs from $q$ in that it misses atom $g_j(\bar y)$. The atom can be deleted with a rewrite only if it is exogenous, or $\exists g_i(\bar x)$, with $\bar x\subseteq\bar y$.
    
    We will reduce responsibility for $q'$ to $q$:
    Take the output tuples of $q'$ and project on $\bar y$. Assign the result to $g_j$. If $g_j$ is exogenous or if $\bar x\subset\bar y$, then tuples from $g_j$ are never picked in the minimum contingency. If $\bar x=\bar y$, no minimum contingency will have the same tuple from $g_i$ and $g_j$. Therefore the contingency tuples from the 2 subgoals can be mapped to one of them, creating a contingency set for $q'$ Therefore, if $q$ is in PTIME, $q'$ is also in PTIME.
\end{proof}

\begin{proof}[\autoref{lemma:weakening}]
    A weakening $q\weakening\hat q$ also results in a new database instance $\hat D$. A weakening does not create any new result tuples, and does not alter the number of tuples of any endogenous non-dominated atoms. Therefore contingencies in $\hat q,\hat D$ only contain tuples from non-weakened atoms making any contingency for $q$ is also a contingency for $\hat q$ and vice versa.
    If any weakening results in a linear query $\hat q$, \specificref{Algorithm}{alg:flowTransform} can be applied to solve it in PTIME, meaning that weakly linear queries are in PTIME.
\end{proof}

\begin{lemma}[Containment]\label{lem:containment}
    If $q$ is final, then $\forall x,y$, $\sg(x)\not\subseteq\sg(y)$, where $\sg(x)$ the set of subgoals of $q$ that contain $x$.
\end{lemma}

\begin{proof}
    We will prove by contradiction. Assume that $\sg(x)\subseteq\sg(y)$. Denote as $q_1$ the rewrite $q\rewrite q[\emptyset /x]$, and $q_2$ the rewrite $q\rewrite q[\emptyset /y]$. Since $q$ is final, both $q_1$ and $q_2$ are weakly linear. That means that there are weakenings, denote them with $W_1$ and $W_2$, of $q_1$ and $q_2$ respectively, that produce linear queries $\hat q_1$ and $\hat q_2$. The weakenings $W_1$ and $W_2$ are sets of dissociation and domination operations to queries $q_1$ and $q_2$. $\hat q_1$ and $\hat q_2$ determine linear orderings $L_1$ and $L_2$ of the subgoals of $q_1$ and $q_2$ respectively. All subgoals from $\hat q_1$ containing some variable $z$ appear consecutive in $L_1$, but that may not be true for $q_1$. Similarly for $\hat q_2$ and $L_2$.
    
    Assume a third query $q_3$ as the rewrite $q_2\rewrite q_2[\emptyset /x]$, or equivalently $q_1\rewrite q_1[\emptyset /y]$. Note that any relation that is dominated in $q_1$ or $q_2$ is also dominated in $q_3$.
    Define the connected components of $q_3$ as the set $\mathcal{C}_3=\set{C_1, C_2,\ldots, C_k}$ of maximum cardinality, such that $\forall i\neq j$, $C_i$ and $C_j$ have no variables in common. We denote with $q_3[W_2]$ the application to $q_3$ of all weakening operations defined in $W_2$ that are valid for $q_3$. For example a dissociation that uses variable $x$ is not valid for $q_3$, as $x\not\in q_3$. Similarly, $g_i[W_2]$ denotes the subgoal $g_i$ after the application of all weakening operations defined in $W_2$ to query $q_2$.
    
    Assume that $q_2$ and $q_3$ have the same connected components, and $\hat q_3 \datarule q_3[W_2]$. Then, $\hat q_3$ is equivalent to the rewrite $\hat q_2\rewrite\hat q_2[\emptyset /x]$, and therefore also linear. 
    
    Even if the connected components are not the same, two components do not share the same variables. Therefore, a dissociation operation that relied on a neighborhood that disappeared with the removal of $x$ would have only served to transfer a variable from one component to another. This obviously does not affect the linearity of the result in the case of $q_3$.    
    This means that $q_3$ is weakly linear, and $\hat q_3 \datarule q_3[W_2]$ offers a linearization.
    
    We will now use the connected components $\mathcal{C}_3$ of $q_3$ to define linear orders for $q_2$ and $q_1$.
    
    We map each $C_i\in\mathcal{C}_3$ to the linear order $L_2$. $C_i$ may appear fragmented (not contiguous) in $L_2$. We call a fragment of a component $C_i$ in a sequence $L$, a maximal length subsequence of $L$, such that every subgoal in the subsequence appears in $C_i$. An example is shown in \autoref{fig:figs_Fig_containment}.

    \begin{figure}[htb]
        \centering
            \includegraphics[scale=0.65]{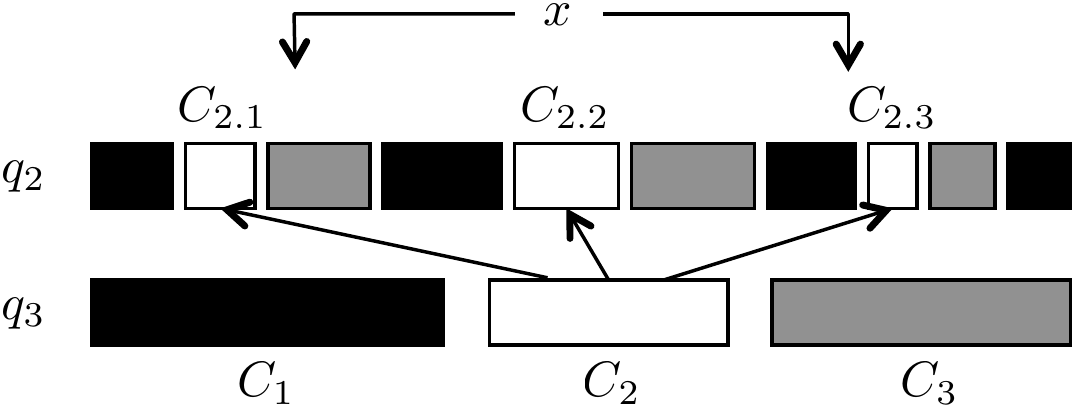}
        \caption{$q_2 = q\rewrite q[\emptyset /y]$ and $q_3 = q_2\rewrite q_2[\emptyset /x]$.}
        \label{fig:figs_Fig_containment}
    \end{figure}

    Assume a component $C_i$ which appears fragmented in $L_2$, and $C_{i.1}$ and $C_{i.2}$ two of its fragments. Let $S$ be the set of subgoals separating $C_{i.1}$ and $C_{i.2}$. All subgoals in $S$ belong to a component other than $C_i$, and therefore share no variables with $C_{i.1}$ or $C_{i.2}$. Assume that there is a subgoal $g_j\in S$ that is not modified by the weakening $W_2$, and therefore contains no variables from $C_i$. If that were the case, then it would be impossible for all variables of $C_i$ to appear consecutively in $L_2$ which is impossible as $L_2$ is a linear order for $\hat q_2$. Therefore, all subgoals of $S$ are modified by $W_2$, and are therefore exogenous (or dominated, and therefore made exogenous).
    Therefore, any connected component of $q_3$ either appears contiguous in $L_2$, or any subgoals separating 2 fragments are exogenous. 
    
    Denote with $L_x$ the subsequence of $L$ containing the ``span'' of variable $x$: $\forall g_i\in L_x$, $x\in g_i[W_2]$, and $\forall g_i\not\in L_x$, $x\not\in g_i[W_2]$. Let $W_x$ be the set of dissociations that propagate $x$ to all the exogenous subgoals adjacent to $L_x$. Let $W_2' = W_2\cup W_x$. Then $\hat q_2'\datarule q_2[W_2']$ is also a linear query, and $L_2$ is a linear order for it.
    
    Since $\hat q_2'$ has $x$ propagated to all adjacent exogenous subgoals, there can only be up to 2 components of $\mathcal{C}_3$ that only partially contain $x$ in $\hat q_2'$. For the rest of the components, either all of their subgoals contain $x$ in $\hat q_2'$ or none of them does. Partition $\mathcal{C}_3$ as follows:
    \begin{align*}
        &C^+_x &=& \set{C_i \,|\, \forall g\in C_i, x\in g[W_2']}\\
        &C^-_x &=& \set{C_i \,|\, \forall g\in C_i, x\not\in g[W_2']}\\
        &C^*_x &=& \set{C_i \,|\, \exists g_r,g_t\in C_i\; s.t.\; x\in g_r[W_2']\; and\; x\not\in g_t[W_2']}
    \end{align*}
    
    As noted, $|C^*_x|\leq 2$. Mapping $\mathcal{C}_3$ to $L_1$ using the same logic, produces sets $C^+_y$, $C^-_y$ and $C^*_y$ for variable $y$ and weakening $W_1'$, which is defined similarly as $W_2'$. Since $\sg(x)\subseteq\sg(y)$, it is $C^+_x\subseteq C^+_y$ and $C^*_x\subseteq C^+_y\cup C^*_y$.
    
    Denote with $C_i[W_j]$ the application of weakening $W_j$ to all subgoals of component $C_i$. Then $\forall C_i\in\mathcal{C}_3$, $C_i[W_1']$ and $C_i[W_2']$ are both linear. Therefore, $q_3[W_1']$ and $q_3[W_2']$ are also linear. Moreover, the components $\mathcal{C}_3$ can be freely rearranged in the linear order, and linearity is always preserved.
   
   Let $C_y^{\neg x} = C^+_y - (C^+_x\cup C^*_x)$, $C_y^l$ and $C_y^r$ the two components in $C_y^*$, and $C_x^l$ and $C_x^r$ the two components in $C_x^*$.

   First, assume $C^*_x\cap C^*_y = \emptyset$, and apply $W'$ to $q$.
   Then it is straightforward that the following is a linear order with respect to all the variables, including $x$ and $y$:
   
  {\small
  \[
  \boxed{C^l_y[W_1']\rule[-1.5mm]{0pt}{.45cm}} \boxed{C^l_x[W_2']\rule[-1.5mm]{0pt}{.45cm}} \boxed{C^+_x[W_2']\rule[-1.5mm]{0pt}{.45cm}} \boxed{C^r_x[W_2']\rule[-1.5mm]{0pt}{.45cm}} \boxed{C^{\neg x}_y[W_1']\rule[-1.5mm]{0pt}{.45cm}} \boxed{C^r_y[W_1']\rule[-1.5mm]{0pt}{.45cm}} \boxed{C^-_y[W_1']\rule[-1.5mm]{0pt}{.45cm}}
  \]}
    
We now examine the case $C^*_x\cap C^*_y\neq\emptyset$. Assume $C_i$ a connected component of $\mathcal{C}_3$, such that $C_i\in C^*_x$ and $C_i\in C^*_y$. $C_i[W_1']$ is in linear order for all variables except possibly for $x$, and $C_i[W_2']$ is in linear order for all variables except possibly for $y$. Let $W_x$ be those weakenings in $W_2'$ that dissociate atoms by adding $x$ to their variable set. Since $\sg(x)\subseteq\sg(y)$, it is possible to also add $y$ to the same atoms. Let $W_{x\rightarrow y}$ be the same weakenings as in $W_x$, but with variable $y$ instead of $x$. Then in $q_{C_i}\datarule C_i[W_2'\cup W_{x\rightarrow y}]$ it also holds that $\sg(x)\subseteq\sg(y)$.

Let $W_y$ be those weakenings in $W_1'$ that dissociate atoms by adding $y$ to their variable set. Then $W_y$ can still be applied to $q_{C_i}$, as atoms that were neighbors in $q_1$ are also neighbors in $q_{C_i}$. Therefore, $q_{C_i}[W_y]$ is linear for all variables including $y$. Thus, the same order of components as the one above, but with $q_{C_i}[W_y]$ replacing $C^l_y[W_1']$ and $C^l_x[W_2']$ (equivalently for $C^r_y[W_1']$ and $C^r_x[W_2']$), will be linear in all variables including $x$ and $y$. But that would mean that $q$ is weakly linear, as there is a weakening for it that is linear, which is a contradiction. Therefore, $\sg(x)\not\subseteq\sg(y)$.
\end{proof}

\begin{lemma}\label{lem:3vars}
    If $q$ is final, then it has exactly 3 variables.
\end{lemma}

\begin{proof}
    First of all, if $q$ has 2 or fewer variables, then it is linear, and therefore cannot be final. So, $q$ must have at least 3 variables. We will assume that $\Var(q)>3$ and prove the lemma by contradiction. We will reach contradiction by applying rewrites to $q$, and showing hardness of the resulting query, which means that $q$ cannot be final. There are two possible cases: either all atoms are unary or binary, or there exists an atom with 3 or more variables. We examine the cases separately.
    
    \introparagraph{Case 1} $\forall g_i\in q$, $\Var(g_i)\leq 2$.\\
    Since $q$ is final, it cannot be linear. For non-linearity $q$ needs to have at least 3 non-dominated atoms.
    Also, since every atom has at most 2 variables, $q$ has to be either cyclic, or have a ``corner point'' in its dual hypergraph like the one shown in \autoref{fig:figs_Fig_cornerPoint}. A corner point in the dual hypergraph is defined as a hyperedge which intersects, but does not contain, at least three other hyperedges. We will refer to these as \emph{branches} of the corner point.

    \begin{figure}[htb]
        \centering
            \includegraphics[scale = 0.3]{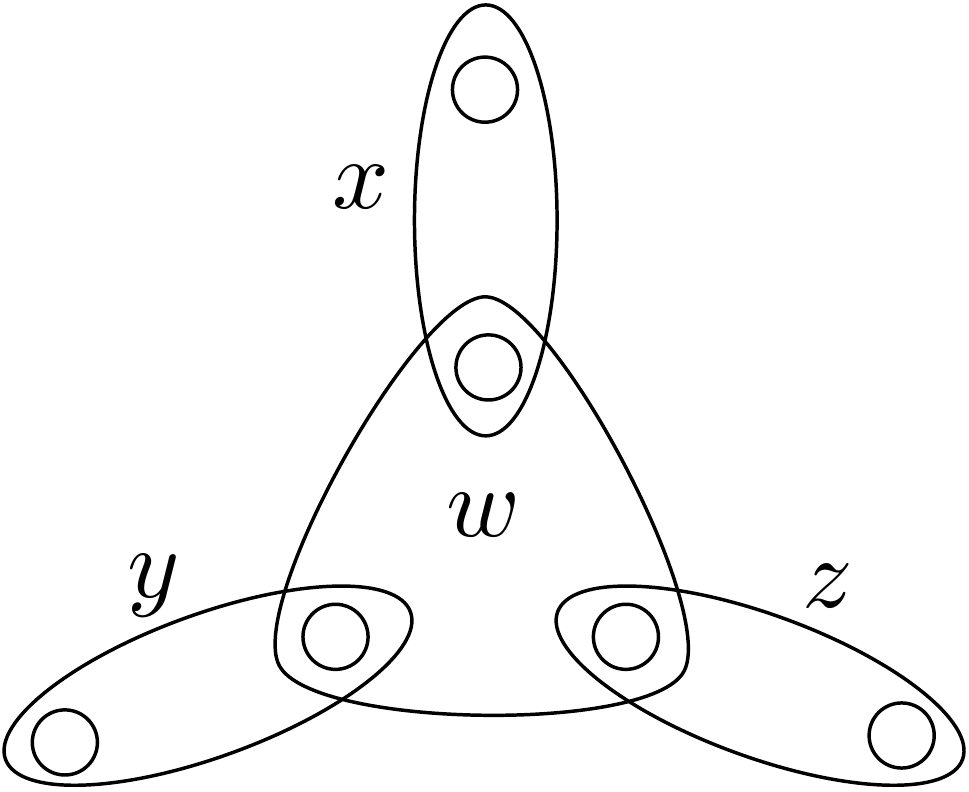}
        \caption{Variable $w$ is a corner point for a query where $\Var(g_i)\leq 2$.}
        \label{fig:figs_Fig_cornerPoint}
    \end{figure}
    
    \underline{Case A:} (The dual hypergraph of $q$ has a corner point)\\
    The three non-dominated atoms have to be on separate branches of the corner point. Assume variable $w$ is the corner point, with the following three corner point branches that together form query $q$:
    \begin{align*}
        &\ldots A(x,u), R_1(x,w_1), R_2(w_1,w_2),\ldots, R_i(w_{i-1},w)\\
        &\ldots B(y,u'), S_1(y,w_1'), S_2(w_1',w_2'),\ldots, S_j(w_{j-1}',w)\\
        &\ldots C(z,u''), T_1(z,w_1''), T_2(w_1'',w_2''),\ldots, T_j(w_{k-1}'',w)
    \end{align*}
    Note that there has to be at least one $R$, one $S$ and one $T$ tuple in order for $w$ to be a corner point.
    W.l.o.g assume that $A$, $B$, $C$ are non-dominated atoms. We can ignore the $u$ variables if the non-dominated atoms are unary. 
    We apply the following rewrites to $q$, $\forall t$:
    \begin{align}
        &q\rewrite q[(w_{t-1}, w_t)/w_t] &&q\rewrite q[(w_{t}, w_{t-1})/w_{t-1}]\label{rewrite1}\\
        &q\rewrite q[(w_{t-1}', w_t')/w_t'] &&q\rewrite q[(w_{t}', w_{t-1}')/w_{t-1}']\label{rewrite2}\\
        &q\rewrite q[(w_{t-1}'', w_t'')/w_t''] &&q\rewrite q[(w_{t}'', w_{t-1}'')/w_{t-1}'']\label{rewrite3}
    \end{align}
    Now all the $R$ tuples contain variables $w_1,\ldots w_{i-1}$, and equivalently for the $S$ and $T$ atoms. We also transfer variables $x$ and $w$ with the following rewrites:
    \begin{align*}
        &q\rewrite q[(x, w_1)/w_1] &&q\rewrite q[(w, w_1)/w_1]\\
        &q\rewrite q[(y, w_1')/w_1'] &&q\rewrite q[(w, w_1')/w_1']\\
        &q\rewrite q[(z, w_1'')/w_1''] &&q\rewrite q[(w, w_1'')/w_1'']
    \end{align*}  
    
    The rewrites transform all of $R_1, R_3,\ldots R_i$ to the same atom $R(x,w_1,w_2,\ldots,w_{i-1},w)$. Equivalently for the $S$ and $T$ atoms. 
    Finally, remove all variables $u$ other than $(x,y,z,w)$ by applying the rewrite $q\rewrite q[\emptyset/u]$. Therefore, after the rewrites the query becomes:
    $$
        q'\datarule A(x),B(y),C(z),R(x,w),S(y,w),T(z,w)
    $$
    We can reduce $h_1^*$ to $q'$ as follows: Atoms $A$, $B$ and $C$ remain unchanged. For each tuple $W(x,y,z)$ we assign a unique value $w$: $(x,y,z,w)$. We get relation $R$ by projecting on $(x,w)$, and similarly for $S$ and $T$. Responsibility of a tuple in $h^*_1$ is the same as the responsibility of the tuple in $q'$. Therefore $q'$ is hard, which means that $q$ cannot be final.

    \underline{Case B:} (The dual hypergraph of $q$ is cyclic)\\
We will distinguish between 4 possibilities for the non-dominated relations: (a) all of them are unary, (b) exactly one is binary, (c) exactly two are binary, and (d) all of them are binary.    

(a) Let $A$, $B$ and $C$ be the non-dominated atoms. Then $q$ is of the form:
\begin{align*}
    &A(x), R_1(x,w_1), R_2(w_1,w_2),\ldots, R_i(w_{i-1},y)\\
    &B(y), S_1(y,w_1'), S_2(w_1',w_2'),\ldots, S_j(w_{j-1}',z)\\
    &C(z), T_1(z,w_1''), T_2(w_1'',w_2''),\ldots, T_j(w_{k-1}'',x)
\end{align*}
We apply the rewrites \eqref{rewrite1},\eqref{rewrite2} and \eqref{rewrite3} from case A. Since we know that $\Var(q)>3$, we are guaranteed that at least one rewrite will happen, as there must exist at least one $w$ variable. The rewrites transform all of $R_1,\ldots R_i$ to the same atom $R(x,w_1,\ldots,w_{i-1},y)$. Equivalently for the $S$ and $T$ atoms. Name the result query $q'$. Query $h^*_3$ can be trivially reduced to $q'$ by setting all the $w$ variables to a constant. Therefore $q'$ is hard, which means that $q$ cannot be final.

(b) Let $A$, $B$ and $C$ be the non-dominated atoms, and $q$ is of the form:
\begin{align*}
    &A(x), R_1(x,w_1), R_2(w_1,w_2),\ldots, R_i(w_{i-1},y)\\
    &B(y), S_1(y,w_1'), S_2(w_1',w_2'),\ldots, S_j(w_{j-1}',z)\\
    &C(z,u), T_1(u,w_1''), T_2(w_1'',w_2''),\ldots, T_j(w_{k-1}'',x)
\end{align*}
We apply the rewrites \eqref{rewrite1},\eqref{rewrite2} and \eqref{rewrite3} from case A. The query then becomes:
\begin{align*}
    &A(x), R(x,w_1,w_2,\ldots,w_{i-1},y)\\
    &B(y), S(y,w_1',w_2',\ldots,w_{i-1}',z)\\
    &C(z,u), T(u,w_1'',w_2'',\ldots,w_{i-1}'',x)
\end{align*}
It is possible that there are not $w$ values, and no rewrites actually occurred. Apply the rewrite $q\rewrite q[(u,z)/u]$. This is guaranteed to occur, due to the given existence of a fourth variable $u$, since $\Var(q)>3$. Further rewrite off all $w$ variables. After these we get query $q'\datarule A(x),B(y),C(z,u),R(x,y),S(y,z),T(z,u,x)$. If we now apply the rewrite $q\rewrite q[\emptyset/u]$ we get $h^*_3$ which is hard. Therefore, $q$ cannot be final.

(c) Let $A$, $B$ and $C$ be the non-dominated atoms, and $q$ is of the form:
\begin{align*}
    &A(x), R_1(x,w_1), R_2(w_1,w_2),\ldots, R_i(w_{i-1},y)\\
    &B(y,v), S_1(v,w_1'), S_2(w_1',w_2'),\ldots, S_j(w_{j-1}',z)\\
    &C(z,u), T_1(u,w_1''), T_2(w_1'',w_2''),\ldots, T_j(w_{k-1}'',x)
\end{align*}
It is possible that $v\equiv z$, but $x$ cannot be the same as any of $(y,v,z,u)$ as in that case $A$ would dominate $B$ or $C$. We apply the rewrites \eqref{rewrite1},\eqref{rewrite2} and \eqref{rewrite3} from case A, and also rewrite off all variables $w$, e.g. $q\rewrite q[\emptyset/w_1]$. The query then becomes: $q'\datarule A(x), B(y,v), C(z,u),  R(x,y), S(v,z), T(u,x)$. Further apply the following rewrites:

{\small
\begin{align*}
    q'  &\rewrite A(x), B(y,v,z), C(z,u), R(x,y), S(v,z), T(u,x)\tag{add z}\\
        &\rewrite A(x), B(y,v,z), C(z,u), R(x,y), S(y,v,z), T(u,x)\tag{add y}\\
        &\rewrite A(x), B(y,z), C(z,u), R(x,y), T(u,x)\tag{delete S and v}\\
        &\rewrite A(x), B(y,z), C(z,u), R(x,y), T(z,u,x)\tag{add z}\\
        &\rewrite A(x,y), B(y,z), C(z,u), R(x,y), T(z,u,x,y)\tag{add y}\\
        &\rewrite A(x,y), B(y,z), C(z,u), T(z,u,x,y)\tag{delete R}\\
        &\rewrite A(x,y), B(y,z), C(z,u,x), T(z,u,x,y)\tag{add x}\\
        &\rewrite A(x,y), B(y,z), C(z,x)\tag{delete T and u}
\end{align*}
}
Note that in atom $S$ has been eliminated, so even if it was $v\equiv z$, in which case there would be no $S$, it would not affect the result; basically the first 3 of the above rewrites would be unnecessary. The last query is $h^*_2$, which is hard, which means that $q$ cannot be final.
    
(d)
Let $A$, $B$ and $C$ be the non-dominated atoms, and $q$ is of the form:
\begin{align*}
    &A(x,t), R_1(t,w_1), R_2(w_1,w_2),\ldots, R_i(w_{i-1},y)\\
    &B(y,v), S_1(v,w_1'), S_2(w_1',w_2'),\ldots, S_j(w_{j-1}',z)\\
    &C(z,u), T_1(u,w_1''), T_2(w_1'',w_2''),\ldots, T_j(w_{k-1}'',x)
\end{align*}
It is possible that $t\equiv y$ or $v\equiv z$ (in which cases there would be no $R$ or $S$ atoms respectively). Again we apply the rewrites \eqref{rewrite1},\eqref{rewrite2} and \eqref{rewrite3} from case A, and also rewrite off all variables $w$, e.g. $q\rewrite q[\emptyset/w_1]$. The query then becomes: 
$$q'\datarule A(x,t), B(y,v), C(z,u),  R(t,y), S(v,z), T(u,x)$$ 

To account for the case $t\equiv y$ or $v\equiv z$, we will eliminate with rewrites the $R$ and $S$ atoms, so the effect of the possible variable equivalence will be eliminated:
{\small
\begin{align*}
    q'  &\rewrite A(x,t), B(y,v,z), C(z,u), R(t,y), S(v,z), T(u,x)\tag{add z}\\
        &\rewrite A(x,t), B(y,v,z), C(z,u), R(t,y), S(y,v,z), T(u,x)\tag{add y}\\
        &\rewrite A(x,t), B(y,z), C(z,u), R(t,y), T(u,x)\tag{delete S and v}\\
        &\rewrite A(x,t,y), B(y,z), C(z,u), R(t,y), T(u,x)\tag{add y}\\
        &\rewrite A(x,t,y), B(y,z), C(z,u), R(x,t,y), T(u,x)\tag{add x}\\
        &\rewrite A(x,y), B(y,z), C(z,u), T(u,x)\tag{delete R and t}\\
\end{align*}
}
With further rewrites $q\rewrite q[(u,z)/u]$ and $q\rewrite q[(x,y)/x]$ we get $q'\datarule A(x,y), B(y,z), C(z,u), T(z,u,x,y)$, which as seen in (c) previously leads to $h^*_2$ with further rewrites. Therefore, $q$ cannot be final.

\introparagraph{Case 2} $\exists g_i\in q$, $\Var(g_i)\geq 3$\\
    Let $R(x,y,z,\ldots)$ be an atom of $q$. Let $S_x$, $S_y$ and $S_z$ be the subsets of subgoals of $q$ that contain variables $x$, $y$ and $z$ respectively. From \autoref{lem:containment} we know that $S_x$, $S_y$ and $S_z$ cannot be subsets of one another. We always know that they intersect due to relation $R$. That means that there are 3 more atoms $A$, $B$ and $C$ in $q$ that contain $x$, $y$, $z$. There are 2 possible choices for $A$, $B$ and $C$ that satisfy the containment requirement.
\begin{figure}[htb]
    \centering
        \includegraphics[scale=0.4]{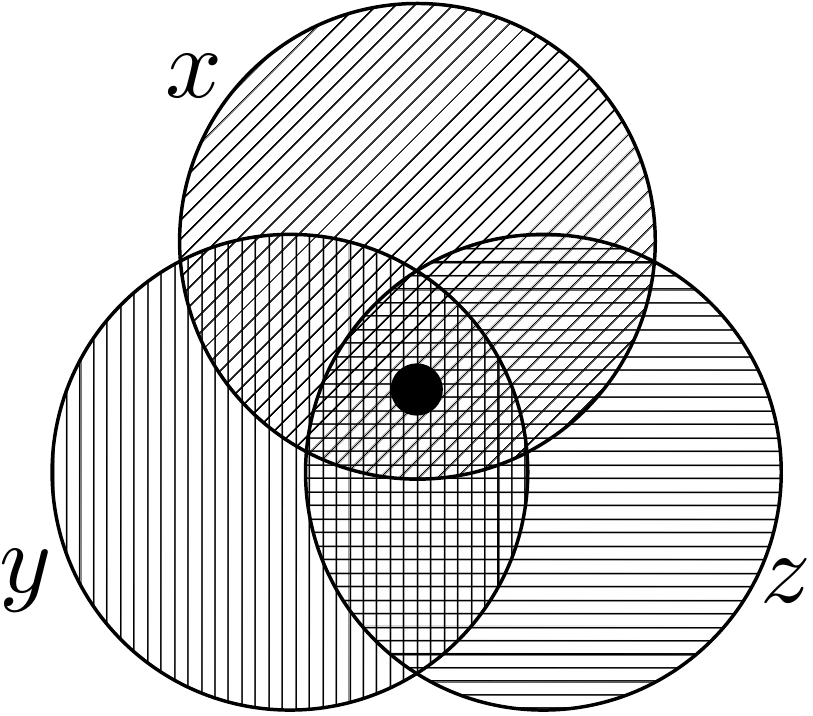}
    \caption{Subgoal $R$ contains all 3 variables. There is only 2 ways to choose $A$, $B$ and $C$ so that $S_x$, $S_y$ and $S_z$ are not subsets of one another.}
    \label{fig:figs_Fig_3vars}
\end{figure}

(a) {\small $A(x,y,\ldots), B(y,z,\ldots), C(z,x,\ldots)$, $z\not\in A$, $x\not\in B$ and $y\not\in C$}\\
For all variables $u$ other than $(x,y,z)$ we rewrite: $q\rewrite q[\emptyset/u]$. Since $\Var(q)>3$, we know that there is at least one such variable $u$.
By further rewriting $q\rewrite q - {R}$, we get $h^*_2$, which is hard, and which means that $q$ is not final.

(b) {\small $A(x,\ldots), B(y,\ldots), C(z,\ldots)$, $y,z\not\in A$, $x,z\not\in B$ and $x,y\not\in C$}\\
For all variables $u$ other than $(x,y,z)$ we rewrite: $q\rewrite q[\emptyset/u]$. Since $\Var(q)>3$, we know that there is at least one such variable $u$.
The result of these rewrites is query $h^*_1$ which is hard, meaning that $q$ is not final.

        Therefore, any final query has exactly 3 variables. 
\end{proof}

\begin{proof}[(\autoref{thm:final})]
    From \autoref{lem:3vars}, since $q$ is final, it has exactly 3 variables and it is not weakly linear. Let $x,y,z$ be the 3 variables. An atom in $q$ can be in one of the following 7 forms: 
    \begin{align*}
            A(x), \;\; B(y), \;\; C(z), \;\; R(x,y), \;\; S(y,z), \;\; T(x,z), \;\; W(x,y,z)
    \end{align*}

    Note that because of the third rewrite rule, and since $q$ is final, we only need to consider queries with at most one atom of each type (e.g. a query containing $R_1(x,y)$ and $R_2(x,y)$ cannot be final). There are 7 possible atom types, and therefore 127 possibilities for queries. We do not need to analyze all of those, as most of them are trivially weakly linear (queries with less that 3 atoms, or less than 3 variables) and therefore cannot be final.
    
    We will show that out of all the possible queries made up as combination of the 7 basic atoms, only $h_1^*$, $h_2^*$, $h_3^*$ are final. The rest are either weakly linear or can be rewritten into the hard queries (and therefore are not final). Note that any query $q$ whose atoms are a subset of the atoms of $h_1^*$ or $h_2^*$ is linear. Therefore, we only need to check supersets of $h_1^*$ and $h_2^*$, and subsets of $h_3^*$. Any query where $A,B,C$ are all endogenous, or $R,S,T$ are all endogenous, is covered by these cases. We finally need to examine queries where at least one unary and one binary atom appear as exogenous.
    For simplicity, from now on we will drop the variable names and just use the relation symbols, in direct correspondence to the atom types mentioned above (e.g. we write $A$ instead of $A(x)$ and $R$ instead of $R(x,y)$). Also, if the endogenous or exogenous state is not explicit, it is assumed that the atom can be in either state.
    
    \underline{Case 1:} supersets of $h_1^*\datarule \en{A},\en{B},\en{C},W$. The only possible relations that can be added to $h_1^*$ from the possible types are the binary relations $R$, $S$ and $T$. For any of them, that are added to $h_1^*$, there exists a singleton relation ($A$, $B$ or $C$) with a subset of their variables. Therefore, we can apply the third rewrite, eg. $q\rewrite q - \set{R}$ to get $h_1^*$. Therefore, any query $q$ over variables $x,y,z$ that is a superset of $h_1^*$ is not final because it can be rewritten to $h_1^*$.
    
    \underline{Case 2:} supersets of $h_2^*\datarule \en{R},\en{S},\en{T}$. The possible atom types that could be added are the ternary relation $W$, or the unary atoms $A$, $B$, and $C$. There are 5 possible cases excluding symmetries (adding $A$ to $h_2^*$ is equivalent to adding $B$ instead). An atom in parentheses means that it may or may not be part of $q$. 
    \begin{enumerate}[label=(\alph*), itemsep=0pt, parsep=1pt, topsep = 3pt]
        \item $q\datarule (\ex{A},)(\ex{B},)(\ex{C},) \en{R},\en{S},\en{T},W$: $W$ (and $\ex{A}, \ex{B}, \ex{C}$) can be eliminated based on the third rewrite leading to $h_2^*$.
        \item $q\datarule \en{A},(\ex{B},)(\ex{C},){R},{S},{T} (,W)$: weakly linear as $A$ dominates $R$ and $T$ (and $W$), and $\ex{B}$ and $\ex{C}$ can be dissociated to $W$.
        \item $q\datarule \en{A},\en{B},(\ex{C},){R},{S},{T} (,W)$: weakly linear as $R$, $S$, $T$ (and $W$) are dominated, and $\ex{C}$ can be dissociated to $W$.
        \item $q\datarule \en{A},\en{B},\en{C},{R},{S},{T}$: this is $h_3^*$.
        \item $q\datarule \en{A},\en{B},\en{C},{R},{S},{T},W$: $W$ can be eliminated based on the third rewrite leading to $h_3^*$.
    \end{enumerate}
    
    \underline{Case 3:} subsets of $h_3^*\datarule \en{A},\en{B},\en{C},{R},{S},{T}$. We only need to consider those where at least one of the $R,S,T$ atoms is exogenous or missing, otherwise they would fall under case 2. We have the following cases (excluding symmetries):
    \begin{enumerate}[label=(\alph*), itemsep=0pt, parsep=1pt, topsep = 3pt]
        \item $q\datarule {A},{B},{C},{R},{S}$: linear, and therefore any of its subsets are also linear.
        \item $q\datarule \en{B},\en{C},{R},{S},\ex{T}$: weakly linear as $R,S,T$ are dominated and can dissociate to $W$.
        \item $q\datarule \en{C},{R},{S},\ex{T}$: $T$  dissociates to $W$ resulting in a linear query. Any subsets would also be linear.
    \end{enumerate}
    
    \underline{Case 4:} at least one exogenous unary and one exogenous binary relation.
    \begin{enumerate}[label=(\alph*), itemsep=0pt, parsep=1pt, topsep = 3pt]
        \item $q\datarule \ex{A},{B},{C},\ex{R},{S},T,W$: $A$ and $R$ dissociate to W, resulting in a linear query. Any subset is also weakly linear.
        \item $q\datarule \ex{A},{B},{C}, {R},\ex{S},T,W$: $A$ and $S$ dissociate to W, resulting in a linear query. Any subset is also weakly linear.
    \end{enumerate}
    
    Therefore, we have shown that any final query has to be one of $h_1^*$, $h_2^*$, $h_3^*$. 
\end{proof}

\section{Other Responsibility Proofs} 
\label{appendix:responsibilityRest}

\begin{proof}[(\autoref{thm:logspaceComp})]
    We will show the result through a series of reductions. We will start by a known \textsc{logspace} complete problem, the \emph{Undirected Graph accessibility Problem} (UGAP): given an undirected graph $G=(V,E)$ and two nodes $a, b \in V$, decide whether there exists a path from a to b.
    
    \underline{BGAP reduction:} We define the \emph{Bipartite Graph Accessibility Problem} (BGAP): given a bipartite graph $(X,Y,E)$ and two nodes $a \in X$, $b \in Y$, decide whether there exists a path from $a$ to $b$.  Here the path is allowed to traverse edges in both directions, from $X$ to $Y$ and from $Y$ to $X$. We will reduce any instance of UGAP to an instance of BGAP as follows:
    
    Given an instance of UGAP as an undirected graph $G = (V,E)$, and nodes $a, b \in V$, construct a bipartite graph with  $X = V$, $Y = E \cup \{c\}$, where $c$ is a new node, and edges are of the form $(x, (x,y))$ and$ (y, (x,y))$, plus one edge $(b,c)$.  Then there exists a path $a\rightarrow b$ in  $G$ iff there exists a path $a \rightarrow c$ in the bipartite graph. Therefore, BGAP is hard for \textsc{logspace}.
    
    \underline{FPMF reduction:} We define the \emph{Four-Partite Max-Flow} problem (FPMF): given a four-partite network $(U, X, Y, V, E)$ where each edge capacity is either 1 or 2, source and target nodes $s$ and $t$ connected to all nodes in $U$ and $V$ respectively with infinite capacities, and a number $k$, decide whether the max-flow is $\geq k$. We reduce BGAP to FPMF as follows:
    
    Given an instance of BGAP as a bipartite graph $(X,Y,E)$ and two nodes $a \in X$, $b \in Y$, construct a 4-partite graph $(U, X, Y, V, E')$ by leaving the $X$ and $Y$ partitions and edges between them unchanged, as they are in the BGAP instance, and set their capacities to 2. Create a $U$-node $xy$ for each edge  $(x,y)\in E$. Each node $xy\in U$ is connected to $x \in X$ with an edge of capacity 1. Symmetrically, the $V$-nodes are $E$, and each node $y \in Y$ is connected to all nodes $xy \in V$, with capacity 1. Finally connect a source node $s$ to all nodes $U$ with infinite capacity, and connect all nodes in $V$ to a target node $t$ also with infinite capacity. The resulting graph has a maximum flow (min-cut) equal to $|E|$: the number of edges between any 2 partitions is exactly equal to $E$, and edges between the $X$ and $Y$ partitions are not chosen in a minimum cut, as they have capacity 2 instead of 1. The maximum flow of $E$ utilizes all $U-X$ and $Y-V$ edges, and a residual flow of 1 is left in all $X-Y$ edges.
    
    Now add to the graph a new node $a'$ in partition $U$ connected with capacity 1 to node $a$ in $X$ and with infinite capacity to the source node. Also add a node $b'$ to partition $V$, connected to node $b$ of partition $Y$ with capacity 1, and to the target node with infinite capacity. The flow in this final graph is $|E|$ iff there is no path between $a$ and $b$ in the BGAP instance, and it is $|E|+1$ iff there is a path between $a$ and $b$. Therefore, BGAP can be solved by computing the maximum flow in the FPMF instance with $k = |E| + 1$.
    
    \underline{Query reduction:} We will reduce FPMF to computing responsibility for query $q$. Let $(X,Y,Z,W,E)$ and number $k$ be an instance of FPMF. For each $(x_i,y_j)$ edge between partitions $X$ and $Y$ create a tuple $R(x_i,1,y_j)$ is the capacity of the edge is 1, and two tuples $R(x_i,1,y_j)$ and $R(x_i,2,y_j)$ if the capacity of the edge is 2. Similarly create relation $S(y,u_2,z)$ based on the $Y-Z$ edges, and relation $T(z,u_3,w)$ based on the $Z-W$ edges. Finally, add tuples $R(x_0,1,y_0)$, $S(y_0,1,z_0)$, and $T(z_0,1,w_0)$, where $x_0$, $y_0$, $z_0$ and $w_0$ are unique new values to the respective domains. The max-flow in the FPMF instance is $\geq k$ iff the responsibility of $R(x_0,1,y_0)$ is $\geq k$. Therefore, responsibility for $q$ is hard for \textsc{logspace}.
\end{proof}

\begin{proof}[(\autoref{prop:selfJoin} Self-Joins)]
This results from a reduction from vertex cover. Given graph $G(V,E)$ as an instance of a vertex cover problem, we construct relations R and S as follows: 
\begin{itemize}[itemsep=0pt, parsep=1pt, topsep = 0pt]
	\item For each vertex $v_i\in V$, create a new tuple $r_i$ with unique value of attribute $x_i$.
	\item For each edge $(v_i,v_j)\in E$, create a new tuple $s_k$, with values $(x,y)=(x_i,x_j)$, where $x_i$, $x_j$ are the values of tuples $r_i$, $r_j$ that correspond to nodes $v_i$, and $v_j$ respectively.
	\item Add tuples $r_0$ with value $x_0$ and $s_0$ with value $(x_0,x_0)$.
\end{itemize} 

The above transformation is polynomial, as we create one tuple per node and one tuple per edge. A vertex cover of size $K$ in $G$ is a contingency of size $K$ for tuple $r_0$ in the database instance: removing from the database all tuples $r_i$ corresponding to the cover leaves no other join result apart from the one due to $r_0$, $s_0$: All other $s_i=(x_i,y_i)\not\equiv s_0$ do not produce a join result, as at least one of $R(x_i)$ or $R(y_i)$ has been removed.

Now assume a contingency $\mathcal{S}$ for tuple $r_0$. If $\mathcal{S}$ contains a tuple $s_i=(x_i,y_i)$, then we can construct a new contingency $\mathcal{S}'=(\mathcal{S}\setminus\{ s_i\})\cup\{ R(x_i)\}$, and $|\mathcal{S}'|\leq|\mathcal{S}|$. Therefore, there exists a minimum contingency $\mathcal{S}$ that contains only $R$-tuples. If $V'$ the set of nodes that corresponds to tuples $r_i\in\mathcal{S}$, then $V'$ is a vertex cover in $G$. If there was an edge left uncovered, then that means that there would be a tuple $S(x_i,y_i)$, such that neither of $R(x_i)$, $R(y_i)$ are in the contingency, which is a contradiction as the join tuple $R(x_i),S(x_i,y_i),R(y_i)$ would then be in the result.
The cover $V'$ is minimal, because $\mathcal{S}$ is minimal.
\end{proof}

\begin{proof}[(\autoref{thm:whyNoResp} \sql{Why-No} Responsibility)]~\!\!\!
	This is a straightforward result based on the observation that the contingency set of a non-answer is bounded by the query size, and is therefore irrelevant to data complexity. In order to make a tuple counterfactual, we need to insert at most $m-1$ tuples to the database, where $m$ the number of query subgoals. 
\end{proof}

\end{document}